\documentclass{lmcs} 

\pdfoutput=1

\usepackage{lastpage}
\renewcommand{\dOi}{14(4:15)2018}
\lmcsheading{}{1--\pageref{LastPage}}{}{}%
{Feb.~09,~2017}{Nov.~16,~2018}{}


\ACMCCS{[{\bf Theory of computation}]: Algebraic Semantics---Axiomatic Semantics}

\usepackage{latexsym,dsfont,stmaryrd}
\usepackage{xifthen}
\usepackage{tikz}
\usetikzlibrary{calc,arrows,patterns} 
\usetikzlibrary{shapes.multipart}
\usetikzlibrary{matrix,arrows} 
  \tikzset{
    commutative diagram/.style 2 args={
    	matrix of math nodes, row sep=#1,column sep=#2,
	text height=1.5ex, text depth=0.25ex},
    commutative diagram/.default={1cm}{1cm}
    }
  \tikzset{    
    skip loop/.style n args={3}{to path={-- ++(0,#1) -| node[pos=0.25,#2] {#3} (\tikztotarget)}},
    cross line/.style={preaction={draw=white, -, line width=6pt}}
  }
  \tikzstyle{labels} = [
  state/.style={rectangle, draw, thin, inner sep=0.8ex, font=\small}, 
  vars/.style={draw=gray, fill=gray, text=white},
  initial/.style={very thick}
  ]
\usepackage{hyperref}

\providecommand*{\ifempty}[3]{\ifthenelse{\isempty{#1}}{#2}{#3}}
\newcommand{\parensmathoper}[2]{\ensuremath{{{#1}}\ifempty{#2}{}{(#2)}}}
\newcommand{\ol}{\overline}

\newcommand{\ind}[1]{\mathds{1}_{#1}}

\newcommand{\e}{\varepsilon}

\newcommand{\coupling}[2]{\Omega(#1, #2)}
\newcommand{\set}[2]{\left\{ #1 \ifempty{#2}{}{\mid #2} \right\}}
\newcommand{\kernel}{\parensmathoper{\ker}}
\newcommand{\dist}[1][]{\mathbf{d}_{#1}}

\newcommand{\norm}[2][]{\| #2 \|_{#1}}
\newcommand{\K}[1][d]{\parensmathoper{\parensmathoper{\mathcal{K}}{#1}}}

\newcommand{\Labels}{\mathcal{L}}
\newcommand{\naturals}{\mathbb{N}}
\newcommand{\prationals}{\mathbb{Q}_{\geq 0}}

\newcommand{\Q}[1][+]{\mathbb{Q}_{#1}}

\newcommand{\M}{\mathcal{M}} 
\newcommand{\N}{\mathcal{N}} 
\newcommand{\D}{\mathcal{D}} 

\newcommand{\OMC}{\mathbf{OMC}} 
\newcommand{\X}{\mathcal{X}} 
\newcommand{\fn}{\parensmathoper{fn}} 

\newcommand{\bisim}[1][]{\sim^{#1}} 

\newcommand{\TT}[2][]{\mathbb{T}\ifempty{#2}{\ifempty{#1}{}{(#1)}}{\ifempty{#1}{(#2)}{(#1,#2)}}}  
\newcommand{\Sub}[1][]{\ifempty{#1}{\mathcal{S}}{\mathcal{S}(#1)}} 
\newcommand{\E}[1][]{\ifempty{#1}{\mathcal{E}}{\mathcal{E}(#1)}} 
\newcommand{\U}[1][]{{\mathcal{U}^{#1}}} 
\newcommand{\KK}[2][]{\mathbb{K}\ifempty{#1}{(#2)}{(#1,#2)}} 
\newcommand{\freemodel}[1][\vdash]{{\mathcal{T}_{#1}}} 

\newcommand{\Refl}{\textsf{Refl}} 
\newcommand{\Symm}{\textsf{Symm}} 
\newcommand{\Triang}{\textsf{Triang}} 
\newcommand{\Max}{\textsf{Max}} 
\newcommand{\Arch}{\textsf{Arch}} 
\newcommand{\Nexp}{\textsf{NExp}} 
\newcommand{\Subst}{\textsf{Subst}} 
\newcommand{\Cut}{\textsf{Cut}} 
\newcommand{\Assum}{\textsf{Assum}} 
\newcommand{\Bone}{\textsf{B1}}
\newcommand{\Btwo}{\textsf{B2}}
\newcommand{\SC}{\textsf{SC}}
\newcommand{\SA}{\textsf{SA}}

\newcommand{\Unfold}{\textsf{Unfold}}
\newcommand{\Unguard}{\textsf{Unguard}}
\newcommand{\Fix}{\textsf{Fix}}
\newcommand{\Cong}{\textsf{Cong}}
\newcommand{\Top}{\textsf{Top}}
\newcommand{\Conv}{\textsf{IB}}
\newcommand{\Dist}{\textsf{Dist}}
\newcommand{\Pref}{\textsf{Pref}}
\newcommand{\Entr}{\textsf{Entr}}
\newcommand{\Distr}{\textsf{Distr}}
\newcommand{\dPref}[1][\lambda]{\textsf{$#1$-Pref}}

\renewcommand{\O}{\Sigma} 
\newcommand{\denot}[1]{\llbracket #1 \rrbracket}
\newcommand{\rec}[2]{\mathop{\textsf{rec}}#1\ifempty{#2}{}{. #2}} 
\newcommand{\pref}[2]{#1. #2} 

\newcommand{\UU}{\mathbb{U}} 
\newcommand{\tauTT}[1][]{\ifempty{#1}{\mu_{\TT{}}}{\mu_{#1}}} 
\newcommand{\PP}{\mathcal{P}}

\newcommand{\A}{\mathcal{A}} 

\theoremstyle{plain} 

\def\ie{{\em i.e.}}

\def\cf{{\em cf.}}

\begin{document}

\title[A Quantitative Deduction System for the Bisimilarity Distance on MCs]{
A Complete Quantitative Deduction System for the Bisimilarity Distance on Markov Chains}
\titlecomment{This paper is an extended version of an earlier conference paper~\cite{BacciBLM16:CONCUR}
presented at CONCUR 2016.}

\author[G.\ Bacci]{Giorgio Bacci}	
\address{Dept.\ of Computer Science, Aalborg University, Denmark}	
\email{\{grbacci,giovbacci,kgl,mardare\}@cs.aau.dk}  
\thanks{Work supported by the EU 7th Framework Programme (FP7/2007-13)
under Grants Agreement nr.318490 (SENSATION), nr.601148 (CASSTING), the 
Sino-Danish Basic Research Center IDEA4CPS funded by Danish National Research 
Foundation and National Science Foundation China, the ASAP Project (4181-00360) funded by the Danish Council for Independent Research, the ERC Advanced Grant LASSO, and the Innovation Fund Denmark center DiCyPS.
}	

\author[G.\ Bacci]{Giovanni Bacci}	

\author[K.\ G.\ Larsen]{Kim G. Larsen}	

\author[R.\ Mardare]{Radu Mardare}	





\begin{abstract}
In this paper we propose a complete axiomatization of the bisimilarity distance of 
Desharnais et al.\ for the class of finite labelled Markov chains. 

Our axiomatization is given in the style of a quantitative extension of equational logic recently proposed by Mardare, Panangaden, and Plotkin (LICS 2016) that uses equality relations $t \equiv_\varepsilon s$ indexed by rationals, expressing that ``$t$ is approximately equal to $s$ up to an error $\varepsilon$''.
Notably, our quantitative deduction system extends in a natural way the equational system for probabilistic bisimilarity given by Stark and Smolka by introducing an axiom for dealing with the Kantorovich distance between probability distributions.

The axiomatization is then used to propose a metric extension of a Kleene's style representation theorem for 
finite labelled Markov chains, that was proposed (in a more general coalgebraic fashion) by Silva et al.\ (Inf.\ Comput.\ 2011).
\end{abstract}

\maketitle

\section{Introduction}

In~\cite{Kleene56}, Kleene presented an algebra of \emph{regular events} with the declared objective of
``\textit{showing that all and only regular events can be represented by [...] finite automata}.''
This fundamental correspondence is known as Kleene's representation theorem for regular languages.
Kleene's approach was essentially equational, but he did not provide a proof of completeness for his equational characterization. The first sound and complete axiomatization for proving equivalence of regular
events is due to Salomaa~\cite{Salomaa66}, later refined by Kozen~\cite{Kozen91}.

The above programme was applied by Milner~\cite{Milner84} to process behaviors and non-deterministic labelled transition system. Milner's algebra of process behaviors consists of a prefix operator (representing the observation of an atomic event), a non-deterministic choice operator (for union of behaviors), and a recursive operator (for the definition of recursive behaviors). He proved an analogue of Kleene's theorem, showing that process behaviors represent all and only finite labelled transition systems up to bisimilarity. 
Milner also provided a sound and complete axiomatization for behavior expressions, with the property that two expressions are provably equal if and only if they represent bisimilar labelled transition systems.

Stark and Smolka~\cite{StarkS00} extended Milner's axiomatization to probabilistic process behaviors,
providing a complete axiomatization for probabilistic bisimilarity on (generative) labelled Markov chains.
The key idea was to represent probabilistic non-determinism by a probabilistic choice 
operator for the convex combination of probabilistic behaviors. Their equational axiomatization 
differs from Milner's only by replacing the semi-lattices axioms for non-deterministic choice by the Stone's barycentric axioms for probabilistic choice.

Similar extensions to Milner's axiomatization have been investigated by several authors.
Amongst them we recall some of the works that have been done on probabilistic systems: 
Bandini and Segala~\cite{BandiniS01} on simple probabilistic automata;
Mislove, Ouakinine, and Worrell~\cite{MisloveOW04} on (fully) non-deterministic probabilistic automata;
Deng and Palamidessi~\cite{DengP07} axiomatizing probabilistic weak-bisimulation and behavioral equivalence on Segala and Lynch's probabilistic automata; and Silva et al.~\cite{SilvaBBR11,Bonsangue13} 
providing a systematic way to generate sound and complete axiomatizations and Kleene's representation theorems thereof, for a wide variety of systems represented as coalgebras.  

The attractiveness towards sound and complete axiomatizations for process behaviors 
originates from the need of being able to reason
about their equivalence in a purely algebraic fashion 
by means of classical deduction of valid equational statements.

Jou and Smolka~\cite{JouS90}, however, observed that for reasoning about the behavior of probabilistic systems (and more in general all type of quantitative systems) a notion of distance is preferable to that of equivalence, since the latter is not robust w.r.t.\ small variations of numerical values.  
This motivated the development of metric-based approximated semantics for probabilistic systems, initiated by Desharnais et al.~\cite{DesharnaisGJP04} on labelled Markov chains and greatly developed and explored by van Breugel, Worrell and others \cite{BreugelW:icalp01,BreugelW06,BreugelSW08,BacciLM:fossacs15,BacciLM:ictac15}. 
It consists in proposing a pseudometric which 
measures the dissimilarities between quantitative behaviors. 
The pseudometric proposed by Desharnais et al.~\cite{DesharnaisGJP04} on labelled Markov chains, 
a.k.a.\ \emph{probabilistic bisimilarity distance}, is defined as the least solution of a functional operator 
on $1$-bounded pseudometrics based on the Kantorovich distance on probability distributions.

The first proposals of sound and complete axiomatizations of behavioral distances are due to Larsen et al.~\cite{LarsenFT11} on weighted transition systems, and D'Argenio et al.~\cite{DArgenioGL14} on probabilistic systems. These approaches, however, are rather specific and based on \textit{ad hoc} assumptions.
Recently, Mardare, Panangaden, and Plotkin~\cite{MardarePP:LICS16} ---with the purpose of developing a general research programme for a quantitative algebraic theory of effects~\cite{PlotkinP01}--- proposed the concept of \emph{quantitative equational theory}. The key idea behind their approach is to use ``quantitative equations'' of the form $t \equiv_\e s$ to be interpreted as ``$t$ is approximately equal to $s$ up to an error $\e$'', for some rational number 
$\e \geq 0$. 
Their main result is that completeness for a quantitative theory always holds on the freely-generated algebra of terms equipped with a metric that is freely-induced by the axioms.
Due to this result, they were able to prove soundeness and completeness theorems for many interesting axiomatizations, such as the Hausdorff metric, the total variation metric, the $p$-Wasserstein metric, and the Kantorovich metric.

In this paper, we contribute to the quest of axiomatizations of behavioral metrics, by proposing a quantitative deduction system in the sense of~\cite{MardarePP:LICS16}, that is proved to be sound and complete w.r.t.\ the probabilistic bisimilarity distance of Desharnais et al. on labelled Markov chains. 
The proposed axiomatization extends Stark and Smolka's one~\cite{StarkS00} with three 
additional axioms ---the last of these borrowed from~\cite{MardarePP:LICS16}--- for expressing: 
(1) $1$-boundedness of the metric; (2) non-expansivity of the prefix operator; and 
(3) the Kantorovich lifting of a distance to probability distributions.

The resulting axiomatization is simpler than the one presented in~\cite{DArgenioGL14} for probabilistic transition systems and it extends~\cite{DArgenioGL14} by allowing recursive behaviors.

An important detail about our axiomatization that should be mentioned since the introduction, 
is that the proposed quantitative deduction system do not fully adhere to the conditions required 
to fit within the quantitative algebraic framework of~\cite{MardarePP:LICS16}. Indeed, one of the
conditions that we do not satisfy is soundness w.r.t.\ non-expansivity for the recursion operator.
To overcome this problem we propose to relax the original notion of quantitative deduction 
system from~\cite{MardarePP:LICS16} by not requiring non-expansivity of the algebraic 
operators. In this way we could not use the 
general proof technique of~\cite{MardarePP:LICS16} to obtain the completeness theorem, and 
we needed to appeal to specific properties of the distance.
The property used to prove completeness is $\omega$-cocontinuity (i.e., preservation 
of \emph{infima} of countable decreasing chains) of the functional operator defining  
the distance. Interestingly, the proof technique proposed in this paper appears to be generic 
on the used functional operator, provided its $\omega$-cocontinuity.
This generality is showed on a specific example, where we prove soundness and 
completeness for another (although similar) quantitative 
deduction system w.r.t.\ the \emph{discounted} bisimilarity distance of Desharnais et al.

As a final result we prove a metric analogue of Kleene's representation theorem for finite 
labelled Markov chains. Specifically, we show that the class of expressible behaviors 
correspond, up to bisimilarity, to the class of finite labelled Markov chains. 
Moreover, if the set of expressions is equipped with 
the pseudometric that is freely-induced by the axioms, this correspondence is metric invariant.
Note that this establishes a stronger correspondence than the usual ``equational''
Kleene's theorem.

\smallskip
This work is an extended version of the conference paper~\cite{BacciBLM16:CONCUR}.
In comparison to~\cite{BacciBLM16:CONCUR}, this paper includes new and more detailed examples, 
improved proofs of soundness and completeness of the axiomatization, and two new results: (i)
a sound and complete axiomatization for the discounted bisimilarity distance 
(Section~\ref{sec:discountaxioms}) and (ii) 
a quantitative Kleene's representation theorem for finite (open) Markov chains 
(Section~\ref{sec:expressibility}).

\subsubsection*{Synopsis} Section~\ref{sec:prelim} introduces the notation and the preliminary basic
concepts. In Section~\ref{sec:QDS} we recall the quantitative equational
framework of~\cite{MardarePP:LICS16}. Section~\ref{sec:AMC} presents the algebra of open 
Markov chains. 
In Section~\ref{sec:axiomatizationbisim} we present a quantitative deduction system (Section~\ref{sec:quantAxioms}) that is proved to be sound (Section~\ref{sec:soundness}) and complete 
(Section~\ref{sec:completeness}) w.r.t.\ the probabilistic bisimilarity distance of Desharnais et al.. 
Section~\ref{sec:discountaxioms} shows how to extend these results to the discounted version of the bisimilarity distance.
In Section~\ref{sec:expressibility} we present a quantitative Kleene's representation theorem
of finite open Markov chains.
Finally, in Section~\ref{sec:conclusion} we conclude and present suggestions for future work.

\section{Preliminaries and Notation} 
\label{sec:prelim}

For $R \subseteq X \times X$ an equivalence relation, we denote by $X /_R$ its quotient set. For two sets $X$ and $Y$, we denote by $X \uplus Y$ their disjoint union and by $[X \to Y]$ (or alternatively, $Y^X$) the set of all functions from $X$ to $Y$.

A \emph{discrete sub-probability distribution} on $X$ is a real-valued function \mbox{$\mu \colon X \to [0,1]$}, such that $\mu(X) \leq 1$,
where, for $E \subseteq X$, $\mu(E) = \sum_{x \in E} \mu(x)$. A sub-probability distribution is a \emph{(full) probability distribution} if $\mu(X) = 1$. The support of $\mu$ is the set of all points with
strictly positive probability, denoted as $\mathit{supp}(\mu) = \set{x \in X}{ \mu(x) > 0}$. 
We denote by $\Delta(X)$ and $\D(X)$ the set of discrete probability and finitely-supported sub-probability distributions on $X$, respectively.

A $1$-bounded \emph{pseudometric} on $X$ is a function $d \colon X \times X \to [0,1]$ such that, for any $x,y,z \in X$, $d(x,x) = 0$, $d(x,y) = d(y,x)$ and $d(x,y) + d(y,z) \geq d(x,z)$; $d$ is a \emph{metric} if, in addition, $d(x,y) = 0$ implies $x = y$. The pair $(X,d)$ is called \emph{(pseudo)metric space}. 
The kernel of a (pseudo)metric $d$ is the set $\kernel{d} = \set{(x,y)}{ d(x,y) = 0}$. 


\section{Quantitative Algebras and their Equational Theories} 
\label{sec:QDS}

We recall the notions of quantitative equational theory and quantitative algebra from~\cite{MardarePP:LICS16}.
 
Let $\Sigma$ be an algebraic signature of function symbols $f \colon n \in \Sigma$ of arity $n \in \naturals$.
Fix a countable set of \emph{metavariables} $X$, ranged over by $x,y,z, \ldots \in X$. We denote by $\TT[\Sigma]{X}$ the set of $\Sigma$-terms freely generated over $X$; terms will be ranged over by $t,s,u,\ldots$
A \emph{substitution of type $\Sigma$} is a function $\sigma\colon X \to \TT[\Sigma]{X}$ that is homomorphically extended to terms as $\sigma(f(t_1, \dots, t_n)) = f(\sigma(t_1), \dots, \sigma(t_n))$; by $\Sub[\Sigma]$ we denote the set of substitutions of type $\Sigma$.

A \emph{quantitative equation of type $\Sigma$} is an expression of the form $t \equiv_\e s$, where
$t,s \in \TT[\Sigma]{X}$ and $\e \in \prationals$. Let $\E[\Sigma]$ denote the set of quantitative equations of type
$\Sigma$. The subsets of $\E[\Sigma]$ will be ranged over by $\Gamma, \Theta, \Pi, \ldots \subseteq \E[\Sigma]$.

Let ${\vdash} \subseteq 2^{\E[\Sigma]} \times \E[\Sigma]$ be a binary relation from the powerset of $\E[\Sigma]$
to $\E[\Sigma]$. We write $\Gamma \vdash t \equiv_\e s$ whenever $(\Gamma,t \equiv_\e s) \in {\vdash}$, and $\vdash t \equiv_\e s$ for $\emptyset \vdash t \equiv_\e s$. We use $\Gamma \vdash \Theta$ as a shorthand to meaning that $\Gamma \vdash t \equiv_\e s$ holds for all $t \equiv_\e s \in \Theta$. The relation $\vdash$ is a \emph{quantitative deduction 
system of type $\Sigma$} if it satisfies the following axioms and rules
\begin{align*} 
(\Refl) \quad 
& \vdash t \equiv_0 t \,, \\
(\Symm) \quad 
& \{t\equiv_\e s\} \vdash s\equiv_\e t \,, \\
(\Triang) \quad 
& \{t \equiv_\e u, u \equiv_{\e'} s \} \vdash t \equiv_{\e+\e'} s \,, \\
(\Max) \quad 
& \{t\equiv_\e s\} \vdash t\equiv_{\e+\e'}s \,, \text{ for all $\e'>0$} \,, \\ 
(\Arch) \quad 
& \{t\equiv_{\e'}s\mid \e'>\e\} \vdash t\equiv_\e s \,, \\
(\Nexp) \quad
& \{t_1=_\e s_1,\ldots,t_n =_\e s_n\} \vdash f(t_1,\dots, t_n) \equiv_\e f(s_1,\dots, s_n) \,, 
\text{ for all $f \colon n \in \Sigma$} \,, \\
(\Subst) \quad
& \text{If $\Gamma \vdash t \equiv_\e s$, then $\sigma(\Gamma) \vdash \sigma(t) \equiv_\e \sigma(s)$, 
for all $\sigma \in \Sub[\Sigma]$} \,, \\
(\Cut) \quad 
& \text{If $\Gamma \vdash \Theta$ and $\Theta \vdash t \equiv_\e s$, then $\Gamma \vdash t \equiv_\e s$} \,, \\
(\Assum) \quad
& \text{If $t \equiv_\e s \in\Gamma$, then $\Gamma \vdash t \equiv_\e s$} \,.
\end{align*}
where $\sigma(\Gamma) = \set{\sigma(t) \equiv_\e \sigma(s)}{ t \equiv_\e s \in \Gamma}$.

\medskip
The rules (\Subst), (\Cut), (\Assum) are the classical deduction rules from equational logic.
The axioms (\Refl), (\Symm), (\Triang) correspond, respectively, to reflexivity, symmetry, and triangular 
inequality for a pseudometric; (\Max) represents inclusion of neighbourhoods of increasing diameter;
(\Arch) is the Archimedean property of the reals w.r.t.\ a decreasing chain of neighbourhoods with 
converging diameters; and (\Nexp) stands for non-expansivity of the algebraic operators $f \in \Sigma$.

\medskip
A \emph{quantitative equational theory} is a set $\U$ of universally quantified \emph{quantitative inferences}, (\ie, expressions of the form 
\begin{equation*}
\{t_1 \equiv_{\e_1} s_1, \dots, t_n \equiv_{\e_n} s_n\} \vdash t \equiv_\e s \,,
\end{equation*}
with a finite set of hypotheses) closed under $\vdash$-deducibility.
A set $\A$ of quantitative inferences is said to axiomatize a quantitative equational theory $\U$, if $\U$ is the smallest quantitative equational theory containing $\A$.
A theory $\U$ is called \emph{inconsistent} if $\vdash x \equiv_0 y \in \U$, for distinct metavariables $x,y \in X$, it is called \emph{consistent} otherwise%
\footnote{Note that for an inconsistent theory $\U$, by $\Subst$, we have $\vdash t \equiv_0 s \in \U$, for all $t,s \in \TT[\Sigma]{X}$.}. 

\medskip
The models of quantitative equational theories are standard $\Sigma$-algebras equipped with a pseudometric,
called \emph{quantitative algebras}. 
\begin{defi}[Quantitative Algebra] \label{def:quantitativealgebra}
A \emph{quantitative $\Sigma$-algebra} is a tuple $\A = (A, \Sigma^{\A}, d^\A)$, consisting of a 
pseudometric space $(A, d^\A)$ and a set $\Sigma^\A = \set{f^\A \colon A^n \to A}{ f \colon n \in \Sigma }$ of \emph{interpretations} for the function symbols in $\Sigma$, required to be non-expansive w.r.t.\ $d^\A$, \ie, for all $1 \leq i \leq n$ and $a_i, b_i \in A$,
$d^\A(a_i, b_i) \geq d^\A( f^\A(a_1, \dots, a_n), f^\A(b_1, \dots, b_n))$.
\end{defi}
 
Morphisms of quantitative algebras are non-expansive homomorphisms.

\medskip
A quantitative algebra $\A = (A,\Sigma^\A,d^\A)$ \emph{satisfies} the quantitative inference $\Gamma \vdash t \equiv_\e s$, written $\Gamma \models_\A t \equiv_\e s$, if for any assignment of the meta-variables $\iota \colon X \to A$,
\begin{align*}
\big( \text{for all $t' \equiv_{\e'} s' \in \Gamma$, } d^\A(\iota(t'),\iota(s')) \leq \e'  \big) 
&& \text{implies} && 
d^\A(\iota(t),\iota(s)) \leq \e \,,
\end{align*}
where, for a term $t \in \TT[\Sigma]{X}$, $\iota(t)$ denotes the homomorphic interpretation of $t$ in $\A$.
A quantitative algebra $\A$ is said to \emph{satisfy} (or is a \emph{model} for) the quantitative theory $\U[]$, 
if whenever $\Gamma \vdash t \equiv_\e s \in \U[]$, then $\Gamma \models_\A t \equiv_\e s$. 
The collection of all models of a theory $\U$ of type $\Sigma$, 
is denoted by $\KK[\Sigma]{\U}$.

In~\cite{MardarePP:LICS16} it is shown that any quantitative theory $\U[]$ has a universal model 
$\freemodel[{\U[]}]$ (the freely generated $\vdash$-model) satisfying exactly those quantitative equations 
belonging to $\U[]$.
Moreover, \cite[Theorem~5.2]{MardarePP:LICS16} proves a completeness theorem for quantitative equational theories $\U$, stating that a quantitative inference is satisfied by all the algebras satisfying $\U$ if and only if it belongs to $\U$.
\begin{thm}[Completeness] For any given quantitative equational theory $\U$ of type $\Sigma$,
\begin{align*}
\big( \text{for any $\A \in \KK[\Sigma]{\U}$, $\Gamma \models_\A t \equiv_\e s$}  \big) 
&& \text{if and only if} && 
\Gamma \vdash t \equiv_\e s \in \U \,.
\end{align*}
\end{thm}

Furthermore, in \cite{MardarePP:LICS16} several interesting examples of quantitative equational theories have been proposed. The one we will focus on later in this paper is the so called interpolative barycentric equational theory (\cf~\S10 in~\cite{MardarePP:LICS16}).


\section{The Algebra of Probabilistic Behaviors} \label{sec:AMC}

Recall from the introduction that the aim of the paper is to study the quantitative algebraic properties of two different behavioral pseudometrics on Markov chains, namely the probabilistic bisimilarity distance and the total variation distance. This will be done by employing the framework of quantitative algebras and their equational theories.
For the moment we will focus only on the purely algebraic part, leaving to later sections the definition of the pseudometrics and the quantitative equational theories.

In this section we present the algebra of open Markov chains. Open Markov chains extend the familiar notion of discrete-time labelled Markov chain with ``open'' states taken from a fixed countable set $\X$ of names ranged over by $X, Y, Z, \ldots \in \X$. Names indicate states at which the behavior of the Markov chain can be extended by substitution of another Markov chain, in a way which will be made precise later.

\subsection{Open Markov Chains}
In what follows we fix a countable set $\Labels$ of labels, ranged over by $a, b, c, \ldots \in \Labels$.
Recall that $\D(M)$ denotes the set of \emph{finitely supported} discrete sub-probability distributions over a set $M$.
\begin{defi}[Open Markov Chain]  \label{def:OMC}
An \emph{open Markov chain} $\M = (M, \tau)$ consists of a set $M$ of \emph{states} and a \emph{transition probability function} $\tau \colon M \to \D((\Labels \times M) \uplus \X)$.
\end{defi}
Intuitively, if $\M$ is in a state $m \in M$, then with probability $\tau(m)(a, n)$ it emits $a \in \Labels$ and moves  to state $n \in M$, or it moves with probability $\tau(m)(X)$ to a name $X \in \X$ without emitting any label. 
A state $m \in M$ with probability zero of emitting any label and moving to any name, is called \emph{terminating}. 
A name $X \in \X$ is said to be \emph{unguarded} in a state $m \in M$, if $\tau(m)(X) > 0$. 
Clearly, open Markov chains where all names are guarded in each state (i.e., for all $m \in M$, $\tau(m)(\X) = 0$) are just standard labelled sub-probabilistic Markov chains.

A \emph{pointed open Markov chain}, denoted by $(\M, m)$, is an open Markov chain $\M = (M, \tau)$ with a distinguished \emph{initial} state $m \in M$. 

\smallskip
Hereafter, we will use $\M = (M, \tau)$, $\N = (N, \theta)$ to range over open Markov chains and $(\M,m)$, 
$(\N, n)$ to range over the set $\OMC$ of pointed open Markov chains. To ease the reading, we will often refer to the constituents of $\M$ and $\N$ implicitly, so that we will try to keep 
this notation consistent as much as possible along the paper.

For the definition of the algebra we will need to consider open Markov chains up to \emph{probabilistic bisimilarity}.
Next we recall its definition, due to Larsen and Skou~\cite{LarsenS89}%
\footnote{The original definition of probabilistic bisimulation was given for classical labelled Markov chains. Definition~\ref{def:bisimulation} is a straightforward generalization of the same concept for open Markov chains.}.
\begin{defi}[Bisimulation] \label{def:bisimulation}
An equivalence relation $R \subseteq M \times M$ is a \emph{bisimulation} on $\M$ if whenever $m \mathrel{R} m'$, then, for all $a \in \Labels$, $X \in \X$ and $C \in M/_R$, 
\begin{enumerate}[label={(\roman*)}]
  \item \label{bisimVars} $\tau(m)(X) = \tau(m')(X)$,
  \item \label{bisimCont} $\tau(m)(\{a\} \times C) = \tau(m')(\{a\} \times C)$.
\end{enumerate}
Two states $m, m' \in M$ are \emph{bisimilar} w.r.t.\ $\M$, written $m \bisim_{\M} m'$, if there exists a bisimulation relation on $\M$ relating them. 
\end{defi}
Intuitively, two states are bisimilar if they have the same probability of \ref{bisimVars} moving to a name $X \in \X$ and \ref{bisimCont} emitting a label $a \in \Labels$ and moving to the same 
bisimilarity class.

\smallskip
In the following we consider bisimilarity between pointed open Markov chains. Two pointed open Markov chains $(\M, m), (\N,n) \in \OMC$ are bisimilar, written $(\M,n) \bisim (\N,n)$, if $m$ and $n$ are bisimilar w.r.t.\ the disjoint union of $\M$ and $\N$ (denoted by $\M \oplus \N$) defined as expected. 
One can readily see that ${\bisim} \subseteq \OMC \times \OMC$ is an equivalence.


\subsection{An Algebra of Open Markov Chains} \label{sec:AOMC}
Next we turn to a simple algebra of pointed Markov chains. The signature of the algebra is defined as follows,
\begin{align*}
  \O = 
    &\set{X \colon 0}{ X \in \X} \cup {} \tag{\sc names} \\
    &\set{\pref{a}{(\cdot)} \colon 1 }{ a \in \Labels } \cup {}  \tag{\sc prefix} \\
    & \set{+_e \colon 2 }{ e \in [0,1]} \cup {}  \tag{\sc probabilistic choice} \\
    & \set{ \rec{X}{} \colon 1}{ X \in \X } \,,
    	\tag{\sc recursion}
\end{align*}
consisting of a constant $X$ for each name in $\X$; a prefix $\pref{a}{\cdot}$ and a recursion $\rec{X}{}$ unary operators, for each $a \in \Labels$ and $X \in \X$; and a probabilistic choice $+_e$ binary operator for each $e \in [0,1]$. For $t \in \TT[\O]{M}$, $\fn{t}$ denotes the set of free names in $t$, where the notions of \emph{free} and \emph{bound name} are defined in the standard way, with $\rec{X}{}$ acting as a binding construct. A term is \emph{closed} if it does not contain any free name. Throughout the paper we consider two terms as syntactically identical if they are identical up to renaming of their bound names ($\alpha$-equivalence). For $t, s_1, \dots, s_n \in \TT[\O]{M}$ and an $n$-vector $\ol{X} = (X_1,\dots, X_n)$ of distinct names, $t[\ol{s}/\ol{X}]$ denotes the simultaneous \emph{capture avoiding substitution} of $X_i$ in $t$ with $s_i$, for $i=1,\dots,n$.  
A name $X$ is \emph{guarded}\footnote{This notion, coincides with the one in~\cite{StarkS00}, though our definition may seem more involved due to the fact that we allow the probabilistic choice operators $+_e$ with $e$ ranging in the closed interval $[0,1]$.} in a term $t$ if every free occurrence of $X$ in $t$ occurs within a context the following forms: $\pref{a}{[\cdot]}$, $s +_1 [\cdot]$, or $[\cdot] +_0 s$.
 
Since from now on we will only refer to terms constructed over the signature $\O$, we will simply write $\TT{M}$ and $\TT{}$, in place of $\TT[\O]{M}$ and $\TT[\O]{\emptyset}$, respectively.

\smallskip
To give the interpretation of the operators in $\O$, we define an operator $\UU$ on open Markov chains, taking $\M$ to the open Markov chain $\UU(\M) = (\TT{M}, \tauTT[\M])$, where the transition probability function
$\tauTT[\M]$ is defined as the least solution over the complete partial order of the set of all functions mapping elements in $\TT{M}$ to a $[0,1]$-valued functions from $(\Labels \times \TT{M}) \uplus \X$, ordered point-wise of the recursive equation
\begin{equation*}
 \tauTT[\M] = \PP_\M\big(\tauTT[\M]\big) \,.
\end{equation*}
The functional operator $\PP_\M$ is defined by structural induction on $\TT{M}$, for arbitrary functions 
$\theta \colon \TT{M} \to [(\Labels \times \TT{M}) \uplus \X \to [0,1]]$, as follows:
\begin{align*}
\begin{aligned}
  \PP_\M(\theta)(m) &= \tau(m) \\
  \PP_\M(\theta)(X) &= \ind{\{X\}}
\end{aligned}
&&
\begin{aligned}
   \PP_\M(\theta)(\pref{a}{t}) &= \ind{\{(a,t)\}} \\
  \PP_\M(\theta)(t +_e s) &= e \theta(t) + (1-e) \theta(s) \\
  \PP_\M(\theta)(\rec{X}{t}) &= \theta(t[\rec{X}{t}/X]) \,,
\end{aligned}
\end{align*}
where $\ind{E}$ denotes the characteristic function of the set $E$.

The existence of the least solution is guaranteed by Tarski's fixed point theorem, since $\PP_\M$
is a monotone operator over the complete partial order defined above.
Notice that, requiring $\tauTT[\M]$ to be the \emph{least} solution is essential for ensuring it to be a proper transition probability function, \ie, that for all $t \in \TT{M}$, $\tauTT[\M](t) \in \D((\Labels \times \TT{M}) \uplus \X)$.
We would also like to point out that, for all $X \in \X$, the above definition renders $\rec{X}{X}$ a terminating state in $\UU(\M)$, that is, $\mu_\M(\rec{X}{X})(\Labels \times \TT{M}) = 0$ and $\mu_\M(\rec{X}{X})(\X) = 0$.

\begin{rem}
The definition of $\mu_\M$ corresponds essentially to the transition probability of the operational semantics of probabilistic processes given by Stark and Smolka in~\cite{StarkS00}. The only difference with their semantics is that the one above is defined over generic terms in $\TT{M}$ rather that just in $\TT{}$. Moreover, our formulation simplifies theirs by skipping the definition of a labelled transition system. 
For a detailed discussion about the well-definition of $\mu_\M$ we refer the interested reader to~\cite{AcetoEI02} and \cite{StarkS00}. 
\end{rem}

\begin{defi}[Universal open Markov chain]
Let $\M$ be an open Markov chain.
The \emph{universal open Markov chain} w.r.t.\ $\M$ is given by $\UU(\M)$.
\end{defi}
The reason why it is called universal will be clarified soon. As for now, just notice that $\UU(\M_\emptyset)$, where $\M_\emptyset = (\emptyset, \tau_\emptyset)$ is the open Markov chain with empty transition function, has $\TT{}$ as the set of states and that its transition probability function corresponds to the one defined in~\cite{StarkS00}. To ease the notation we will denote $\UU(\M_\emptyset)$ as $\UU = (\TT{}, \tauTT{})$.

\medskip
Finally we are ready to define the $\O$-algebra of pointed open Markov chains.

For arbitrary pointed open Markov chains $(\M, m)$, $(\N, n)$ and $n$-ary operator $f \in \O$, define 
$f^\textsf{omc} \colon \OMC^n \to \OMC \in \Sigma^\textsf{omc}$ as follows:
\begin{align*}
  X^\textsf{omc} &= (\UU, X) \,, \\
  (\pref{a}{(\M,m)})^\textsf{omc} &= (\UU(\M), \pref{a}{m}) \,,
\\
  (\M,m) \mathbin{+_e^\textsf{omc}} (\N, n) &= (\UU(\M \oplus \N), m +_e n) \,, \\
  (\rec{X}{(\M,m)})^\textsf{omc} &= (\UU(\M^*_{X,m}), \rec{X}{m}) \,,
\end{align*}
where, for $\M = (M,\tau)$, $\M^*_{X,m}$ denotes the open Markov chain $(M, \tau^*_{X,m})$ with transition function defined, for all $m' \in M$ and $E \subseteq (\Labels \times M) \uplus \X$, as
\begin{equation*}
\tau^*_{X,m}(m')(E) = \tau(m')(X) \tau(m)(E \setminus \{X \}) + \tau(m')(X^c) \tau(m')(E \setminus \{X\}) \,.
\end{equation*}
where $X^c = ((\Labels \times M) \uplus \X)\setminus \{X\}$.
Intuitively, $\tau^*_{X,m}$ modifies $\tau$ by removing the name $X \in \X$ from the support of $\tau(m')$ and replacing it with the probabilistic behavior of $m$.

\begin{defi} \label{def:aOMC}
The \emph{algebra of open pointed Markov chains} is $(\OMC, \O^\textsf{omc})$.
\end{defi}


The (initial) semantics of terms $t \in \TT{}$ as pointed open Markov chains is given via the 
$\Sigma$-homomorphism of algebras $\denot{\cdot} \colon \TT{} \to \OMC$, defined by induction 
on terms as follows
\begin{align*}
\begin{aligned}
\denot{X} &= X^\textsf{omc} \, \\
\denot{\pref{a}{t}} &= (\pref{a}{\denot{t}})^\textsf{omc} \,,
\end{aligned}
&&
\begin{aligned}
\denot{t +_e s} &= \denot{t} \mathbin{+_e^\textsf{omc}} \denot{s} \,, \\
\denot{\rec{X}{t}} &= (\rec{X}{\denot{t}})^\textsf{omc} \,.
\end{aligned}
\tag{\sc semantics}
\end{align*}

\begin{exa} \label{ex:semanticsTerms}
To clarify the semantics of $\O$-terms, we show a step-by-step construction 
of the pointed open Markov chain $\denot{\rec{X}{(\pref{a}{X} +_{\frac{1}{2}} Z)}}$.
We start by giving the semantics of the sub-terms $Z$, $\pref{a}{X}$, and 
$\pref{a}{X} +_{\frac{1}{2}} Z$.
\def\skiph{1}
\def\skipv{1.3}
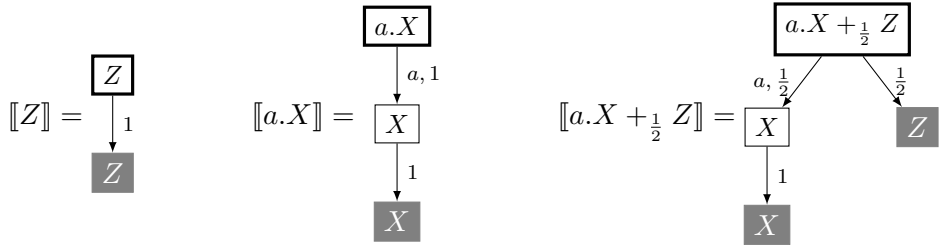
\begin{align*}
\denot{Z} = 
\tikz[labels, baseline={(current bounding box.center)}]{ 
  \draw (0,0) node[state, initial] (z) {$Z$}; 
  \draw ($(z)+(down:\skipv)$) node[vars] (zvar) {$Z$};
  \path[-latex, font=\scriptsize]
    (z) edge node[right] {$1$} (zvar);
}
&&
\denot{a.X} = 
\tikz[labels, baseline={(current bounding box.center)}]{ 
  \draw (0,0) node[state, initial] (aX) {$\pref{a}{X}$}; 
  \draw ($(aX)+(down:\skipv)$) node[state] (x) {$X$};
  \draw ($(x)+(down:\skipv)$) node[vars] (xvar) {$X$};
  \path[-latex, font=\scriptsize]
    (aX) edge node[right] {$a,1$} (x)
    (x) edge node[right] {$1$} (xvar);
}
&&
\denot{\pref{a}{X} +_{\frac{1}{2}} Z} = 
\tikz[labels, baseline={(current bounding box.center)}]{ 
  \draw (0,0) node[state, initial] (choice) {$\pref{a}{X} +_{\frac{1}{2}} Z$}; 
  \draw ($(choice)+(down:\skipv)+(left:\skiph)$) node[state] (x) {$X$};
  \draw ($(choice)+(down:\skipv)+(right:\skiph)$) node[vars] (zvar) {$Z$};
  \draw ($(x)+(down:\skipv)$) node[vars] (xvar) {$X$};
  \path[-latex, font=\scriptsize]
    (choice) edge node[left] {$a,\frac{1}{2}$} (x)
    (choice) edge node[right] {$\frac{1}{2}$} (zvar)
    (x) edge node[right] {$1$} (xvar);
}
\end{align*}
Note that by the definition of the semantic, the pointed Markov chains associated with each term
has the set of terms $\TT{}$ as set of states (hence a countable state space). For the sake of readability, the pointed Markov chains above are presented using the usual graphical representation where only the states reachable from the initial one are shown. White-colored nodes represent the states; grey-colored ones names; and initial states are marked in bold. 

The semantics of $\rec{X}{(\pref{a}{X} +_{\frac{1}{2}} Z)}$ is obtained from 
$\denot{\pref{a}{X} +_{\frac{1}{2}} Z}$ as 
follows:
\begin{align*}
\denot{\rec{X}{\pref{a}{X} +_{\frac{1}{2}} Z}} = 
(\rec{X}{\denot{\pref{a}{X} +_{\frac{1}{2}} Z}})^\textsf{omc} =
\tikz[labels, baseline={(current bounding box.center)}]{ 
  \draw (0,0) node[state, initial] (rec) {$\rec{X}{(\pref{a}{X} +_{\frac{1}{2}} Z)}$}; 
  \draw ($(rec)+(down:\skipv)+(left:\skiph)$) node[state] (x) {$X$};
  \draw ($(rec)+(down:\skipv)+(right:\skiph)$) node[vars] (zvar) {$Z$};
  \path[-latex, font=\scriptsize]
    (rec) edge node[left] {$a,\frac{1}{2}$} (x)
    (rec) edge node[right] {$\frac{1}{2}$} (zvar)
    (x) edge[loop below] node[left] {$a,\frac{1}{2}$} (x)
    (x) edge node[below] {$\frac{1}{2}$} (zvar);
}
\quad\bisim\quad
\tikz[labels, baseline={([yshift=-1.7ex]current bounding box.center)}]{ 
  \draw (0,0) node[state, initial] (m) {$m$}; 
  \draw ($(m)+(down:\skipv)$) node[vars] (zvar) {$Z$};
  \path[-latex, font=\scriptsize]
    (m) edge node[right] {$\frac{1}{2}$} (zvar)
    (m) edge[loop right] node[right] {$a,\frac{1}{2}$} (m);
}
\end{align*}
It is important remarking that by definition of the interpretation of the recursion operator,
for each $t \in \TT{M}$, the states $\rec{X}{t}$ and $X$ are always bisimilar in 
$\denot{\rec{X}{t}}$ (\cf\ picture). 
Hence, the semantics of $\rec{X}{\pref{a}{X} +_{\frac{1}{2}} Z}$ is bisimilar to the pointed 
open Markov chain with initial state $m$ depicted on the right hand side of 
$\denot{\rec{X}{\pref{a}{X} +_{\frac{1}{2}} Z}}$.
\qed
\end{exa}

The next result states that it is totally equivalent to reason about the equivalence of the behavior of 
$\denot{t}$ and $\denot{s}$ by just considering bisimilarity between the corresponding states $t$ and 
$s$ in the universal open Markov chain $\UU$.
\begin{thm}[Universality] \label{th:universal}
For all $t \in \TT{}$, $\denot{t} \bisim (\UU, t)$.
\end{thm}
\begin{proof}[Proof (sketch)] The proof of $\denot{t} \bisim (\UU, t)$ is by induction on $t$. The base case is trivial. The cases for the prefix and probabilistic choice operations are completely routine from the definition of the interpretations and the operator $\UU \colon \OMC \to \OMC$ (in each case a bisimulation can be constructed from those given by the inductive hypothesis).
The only nontrivial case is when $t = \rec{X}{t'}$. The proof carries over in two steps. First one shows that
$(\UU, \rec{X}{t'}) \bisim (\rec{X}{(\UU,t')})^\textsf{omc}$; then, by using the inductive hypothesis $\denot{t'} \bisim (\UU, t')$, that $(\rec{X}{(\UU,t')})^\textsf{omc} \bisim (\rec{X}{\denot{t'}})^\textsf{omc}$. Since $\denot{\rec{X}{t'}} = (\rec{X}{\denot{t'}})^\textsf{omc}$, by transitivity of the bisimilarity relation $\denot{\rec{X}{t'}} \bisim (\UU, \rec{X}{t'})$.
\end{proof}

\begin{rem} \label{rmk:souness&completenessStarkSmolka}
We already noted that the universal open Markov chain $\UU$ corresponds to the operational semantics of probabilistic expressions given by Stark and Smolka~\cite{StarkS00}. In the light of Theorem~\ref{th:universal}, the soundness and completeness results for axiomatic equational system w.r.t.\ probabilistic bisimilarity over probabilistic expressions given in~\cite{StarkS00}, can be moved without further efforts to the class of open Markov chains of the form $\denot{t}$.
\end{rem}


\section{Axiomatization of the Bisimilarity Distance} 
\label{sec:axiomatizationbisim}

In this section we present the probabilistic bisimilarity distance of Desharnais et al.~\cite{DesharnaisGJP04}, that we use to define a quantitative algebra of open Markov chains.
After that, we define a quantitative deduction system which we prove to be sound and complete w.r.t.\ 
probabilistic bisimilarity distance.

We will see that for obtaining these results we cannot directly use the quantitative algebraic
framework of Mardare, Panangaden, and Plotkin~\cite{MardarePP:LICS16} recalled in 
Section~\ref{sec:QDS}, because the recursive operator does not satisfy the conditions 
required to applying their general equational theory. The divergences from~\cite{MardarePP:LICS16} in our development of a quantitative equational theory for the bisimilarity distance over open 
Markov chains will be explained at length as soon as we introduce them.

\subsection{The Probabilistic Bisimilarity Distance}

The notion of probabilistic bisimilarity can be lifted to a pseudometric by means of a straightforward extension to open Markov chains of the probabilistic bisimilarity distance of Desharnais et al.~\cite{DesharnaisGJP04}. 
For more details about its original definition and properties, we refer the interested reader to~\cite{DesharnaisGJP04,BreugelW:icalp01}. 

This distance is based on the \emph{Kantorovich (pseudo)metric} between discrete probability distributions $\mu, \nu \in \Delta(A)$ over a finite set $A$ with underlying (pseudo)metric $d$, defined as  
\begin{equation*}
 \K[d]{\mu,\nu} = \min \set{\sum_{x,y \in A}d(x,y) \cdot \omega(x,y)}{ \omega \in \coupling{\mu}{\nu} } \,.
 \tag{\sc Kantorovich metric}
\end{equation*}
where $\coupling{\mu}{\nu}$ denotes the set of \emph{couplings} for $(\mu, \nu)$, \ie, distributions $\omega \in \Delta(A \times A)$ on the cartesian product $A \times A$ such that, for all $E \subseteq A$, $\omega(E \times A) = \mu(E)$ and $\omega(A \times E) = \nu(E)$.

\begin{rem}
The definition of $\K{}$ is tailored on probability distributions, whereas later we will use it on sub-probability distributions. In these situations one usually interpret sub-probabilities $\mu$ in $A$ as full-probabilities 
$\mu^*$ in $A_\bot$ (i.e., $A$ extended with a bottom element $\bot$ that is assumed to be at maximum distance from all elements $a \in A$) uniquely defined as $\mu^*(E) = \mu(E)$, for all $E \subseteq A$, and $\mu^*(\bot) = 1 - \mu(A)$.
\end{rem}

\begin{defi}[Bisimilarity Distance] \label{def:bisimdist}
Let $\M = (M,\tau)$ be an open Markov chain. The \emph{probabilistic bisimilarity pseudometric} 
$\dist[\M] \colon M \times M \to [0,1]$ on $\M$ is the least fixed-point of the following functional 
operator on $1$-bounded pseudometrics (ordered point-wise),
\begin{equation*}
  \Psi_\M(d)(m,m') =  \K[\Lambda(d)]{\tau^*(m),\tau^*(m')}
  \tag{\sc Kantorovich Operator}
\end{equation*}
where $\Lambda(d)$ is the greatest 1-bounded pseudometric on $\big((\Labels \times M) \uplus \X \big)_\bot$ such that, for all $a \in \Labels$ and $m,m' \in M$, $\Lambda(d)((a,m),(a,m')) = d(m,m')$.
\end{defi}

The well definednees of $\dist[\M]$ is ensured by monotonicity of $\Psi_\M$ (Lemma~\ref{lem:PsiContinuous}) and Knaster-Tarski fixed-point theorem once it is noticed that the set of $1$-bounded pseudometrics with point-wise order $d \sqsubseteq d'$ iff $d(m,n) \leq d'(m,n)$, for all $m,n \in M$, is a complete partial order.

Hereafter, whenever $\M$ is clear from the context we will simply write $\dist$ and $\Psi$ in place of $\dist[\M]$ and $\Psi_\M$, respectively.

\begin{exa} \label{ex:bisimidist}
To better understand how the functional operator $\Psi$ works, we look at a simple 
example of computation of the probabilistic bisimilarity distance $\dist$ between two states.
Consider the open Markov chain and coupling $\omega$ for the 
transition probability distribution $(\tau(m),\tau(n))$ of the states $m$ and $n$ given as follows:
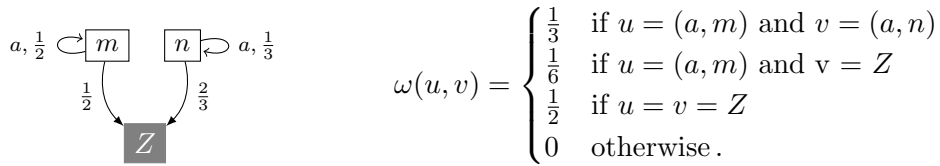
\def\skiph{1}
\def\skipv{1.3}
\begin{align*}
\tikz[labels, baseline={(current bounding box.center)}]{ 
  \draw (0,0) node[state] (m) {$m$}; 
  \draw ($(m)+(right:\skiph)$) node[state] (n) {$n$};
  \draw ($(m)!0.5!(n)+(down:\skipv)$) node[vars] (zvar) {$Z$};
  \path[-latex, font=\scriptsize]
    (m) edge[bend right] node[left] {$\frac{1}{2}$} (zvar)
    (m) edge[loop left] node[left] {$a,\frac{1}{2}$} (m)
    (n) edge[bend left] node[right] {$\frac{2}{3}$} (zvar)
    (n) edge[loop right] node[right] {$a,\frac{1}{3}$} (n);
}
&&
\begin{aligned}
  \omega(u,v) = 
  \begin{cases}
    \frac{1}{3} & \text{if $u = (a,m)$ and $v = (a,n)$} \\
    \frac{1}{6} & \text{if $u = (a,m)$ and v = $Z$} \\
    \frac{1}{2} & \text{if $u = v = Z$ } \\
    0 & \text{otherwise} \,.
  \end{cases}
\end{aligned}
\end{align*}
By definition, $\dist(m,n) = \Psi(\dist)(m,n) = \K[\Lambda(\dist)]{\tau^*(m),\tau^*(n)}$. 
Since $\tau(m)$ and $\tau(n)$ are fully-probability distributions, we apply the direct definition
of Kantorovich distance to get 
\begin{align*}
  \dist(m,n) 
  &= \K[\Lambda(\dist)]{\tau(m),\tau(n)} \\
  &\leq \frac{1}{3} \Lambda(\dist)((a,m), (a,n)) + 
  	\frac{1}{6} \Lambda(\dist)((a,m), Z) +
	\frac{1}{3} \Lambda(\dist)(Z, Z) \,.
\end{align*}
By definition of $\Lambda$, we have that $\Lambda(\dist)((a,m), (a,n)) = \dist(m,n)$, 
$\Lambda(\dist)((a,m), Z) = 1$, and $\Lambda(\dist)(Z, Z) = 0$. 
 
By an easy analysis of the inequality obtained above, one may readily notice that the coupling
$\coupling{\tau(m)}{\tau(n)}$ that minimizes the distance between $m$ and $n$ is the one maximizing
the probability mass on the pairs $((a,m),(a,m))$ and $(Z,Z)$. Since $\omega$ is already 
doing so, we have that $\dist(m,n) = \frac{1}{3} \dist(m,n) + \frac{1}{6}$, thus the distance is $\dist(m,n) = \frac{1}{4}$.
\qed
\end{exa}

By the next lemma and Kleene fixed-point theorem, the bisimilarity distance can be alternatively 
characterized as $\dist = \bigsqcup_{n \in \naturals} \Psi^n(\mathbf{0})$, where $\mathbf{0}$ is the bottom element of the set of $1$-bounded pseudoemetrics ordered point-wise (i.e., the constant $0$ pseudometric).
\begin{lem} \label{lem:PsiContinuous}
$\Psi$ is monotonic and $\omega$-continuous, \ie, for any countable increasing sequence 
$d_0 \sqsubseteq d_1 \sqsubseteq d_2 \sqsubseteq \dots $, it holds $\bigsqcup_{i \in \naturals} \Psi(d_i) = \Psi(\bigsqcup_{i \in \naturals} d_i)$.
\end{lem}
\begin{proof}
Monotonicity of $\Psi$ follows from the monotonicity of $\K[]{}$ and $\Lambda$. 
$\omega$-continuity follows from \cite[Theorem~1]{Breugel12} by showing that $\Psi$ is non-expansive, \ie, for all $d,d' \colon M \times M \to [0,1]$, $\norm{\Psi(d') - \Psi(d)} \leq \norm{d' - d}$, where $\norm{f} = \sup_{x} |f(x)|$ is the supremum norm. It suffices to prove that for all $d \sqsubseteq d'$ and $m, m' \in M$, $\Psi(d')(m,m') - \Psi(d)(m,m') \leq \norm{d' - d}$:
\begin{align*}
&\Psi(d')(m,m') - \Psi(d)(m,m') \\
&= \K[\Lambda(d')]{\tau^*(m),\tau^*(m')} - \K[\Lambda(d)]{\tau^*(m),\tau^*(m')}  \tag{by def.\ $\Psi$} \\
\intertext{by choosing $\omega \in \coupling{\tau^*(m)}{\tau^*(m')}$ so that $\K[\Lambda(d)]{\tau^*(m),\tau^*(m')} = \sum_{x,y} \Lambda(d)(x,y) \cdot \omega(x,y)$,}
&= \textstyle \K[\Lambda(d')]{\tau^*(m),\tau^*(m')} - \sum_{x,y} \Lambda(d)(x,y) \cdot \omega(x,y) \\
& \leq \textstyle \sum_{x,y} (\Lambda(d')(x,y) \cdot \omega(x,y)) - \sum_{x,y} (\Lambda(d)(x,y) \cdot \omega(x,y))\tag{by def.\ of $\K[\Lambda(d')]{}$} \\
&= \textstyle \sum_{x,y} (\Lambda(d')(x,y) - \Lambda(d)(x,y)) \cdot \omega(x,y) \tag{by linearity} \\
\intertext{and since,  for all $(x, y) \notin E = \set{((a,n),(a,n'))}{ a\in \Labels, n,n' \in M}$, $\Lambda(d')(x,y) = \Lambda(d)(x,y)$,}
&= \textstyle \sum_{(x,y) \in E} (\Lambda(d')(x,y) - \Lambda(d)(x,y)) \cdot \omega(x,y) \\
&\leq \textstyle  \sum_{(x,y) \in E} \norm{d' - d} \cdot \omega(x,y) \tag{by def.\ $\Lambda$} \\
&\leq \norm{d' - d} \,. \tag{by linearity and $\omega(E) \leq 1$}
\end{align*}
\end{proof}

Next we show that $\dist$ is indeed a lifting of the probabilistic bisimilarity to pseudometrics.
\begin{lem} \label{lem:bisimdist}
$\dist(m,m') = 0$ iff $m \bisim m'$.
\end{lem}
\begin{proof}
We prove the two implications separately.
($\Leftarrow$) It suffices to show that the relation $R = \set{ (m,m') }{ \dist(m,m') = 0}$ (\ie, $\kernel{\dist}$) is a bisimulation. Clearly, $R$ is an equivalence, and also $\kernel{\Lambda(d)}$ is so. 
Assume $(m,m') \in R$. By definition of $\Psi$, we have that $\K[\Lambda(\dist)]{\tau^*(m),\tau^*(m')} = 0$. By \cite[Lemma 3.1]{FernsPP04}, for all $\kernel{\Lambda(d)}$-equivalence classes
$D \subseteq ((\Labels \times M) \uplus \X)_\bot$, $\tau^*(m)(D) = \tau^*(m')(D)$. By definition of $\Lambda$, this implies that, for all $a \in \Labels$, $X \in \X$ and $C \in M/_R$, $\tau(m)(X) = \tau(m')(X)$ and, moreover, $\tau(m)(\{a\} \times C) = \tau(m')(\{a\} \times C)$.
($\Rightarrow$) 
Let $R \subseteq M \times M$ be a bisimulation on $\M$, and define $d_R \colon M \times M \to [0,1]$ by $d_R(m,m') = 0$ if $(m,m') \in R$ and $d_R(m,m') = 1$ otherwise. We show that $\Psi(d_R) \sqsubseteq d_R$. If $(m,m') \notin R$, then $d_R(m,m') = 1 \geq \Psi(d_R)(m,m')$. If 
$(m,m') \in R$, then for all $a \in \Labels$, $X \in \X$ and $C \in M/_R$, $\tau(m)(X) = \tau(m')(X)$,
$\tau(m)(\{a\} \times C) = \tau(m')(\{a\} \times C)$.
This implies that for all $\kernel{\Lambda(d_R)}$-equivalence class
$D \subseteq ((\Labels \times M) \uplus \X)_\bot$, $\tau^*(m)(D) = \tau^*(m')(D)$. 
By \cite[Lemma 3.1]{FernsPP04}, we have $\K[\Lambda(d_R)]{\tau^*(m),\tau^*(m')} = 0$. This implies that $\Psi(d_R) \sqsubseteq d_R$. Since $\bisim$ is a bisimulation, $\Psi(d_\sim) \sqsubseteq d_\sim$, so that, by Tarski's fixed point theorem, $\dist \sqsubseteq d_\sim$. By definition of $d_\sim$ and $\dist \sqsubseteq d_\sim$, $m \bisim m'$ implies $\dist(m,m') = 0$.
\end{proof}

The bisimilarity distance can alternatively be obtained as $\dist = \bigsqcap_{k \in \naturals} \Tilde{\Psi}^k(\mathbf{1})$, \ie, as the 
$\omega$-limit of the decreasing sequence $\mathbf{1} \sqsupseteq \Tilde{\Psi}(\mathbf{1}) \sqsupseteq \Tilde{\Psi}^2(\mathbf{1}) \sqsupseteq \dots$ of the operator
\begin{equation*}
  \Tilde{\Psi}(d)(m,m') =
  \begin{cases}
    0 &\text{if $m \bisim m'$,} \\
    \Psi(d)(m,m') &\text{otherwise.}
  \end{cases}
\end{equation*}
where $\mathbf{1}$ is the top element of the set of $1$-bounded pseudometrics ordered point-wise
(i.e., the pseudometric that assigns distance $1$ to all distinct elements).

\begin{lem} \label{lem:uniquefix}
$\Tilde{\Psi}$ is monotone and $\omega$-cocontinuous, \ie, for any countable decreasing sequence $d_0 \sqsupseteq d_1 \sqsupseteq d_2 \sqsupseteq \dots $, it holds $\bigsqcap_{i \in \naturals} \Psi(d_i) = \Psi(\bigsqcap_{i \in \naturals} d_i)$. Moreover, $\dist = \bigsqcap_{k \in \naturals} \Tilde{\Psi}^k(\mathbf{1})$.
\end{lem}
\begin{proof}
Monotonicity and $\omega$-cocontinuity follow similarly to Lemma~\ref{lem:PsiContinuous} and \cite[Theorem~1]{Breugel12}.
By $\omega$-cocontinuity $\bigsqcap_{k \in \naturals} \Tilde{\Psi}^k(\mathbf{1})$ is a fixed point. 
By Lemma~\ref{lem:bisimdist} and $\dist = \Psi(\dist)$, also $\dist$ is a fixed point of $\Tilde{\Psi}$.
We show that they coincide by proving that $\Tilde{\Psi}$ has a unique fixed point.

Assume that $\Tilde{\Psi}$ has two fixed points $d$ and $d'$ such that $d \sqsubset d'$. Define
$R \subseteq M \times M$ as $m \mathrel{R} m'$ iff $d'(m,m') - d(m,m') = \norm{d' - d}$. 
By the assumption made on $d$ and $d'$ we have that $\norm{d' - d} > 0$ and ${R} \cap {\bisim} = \emptyset$. Consider arbitrary $m,m' \in M$ such that $m \mathrel{R} m'$, then
\begin{align*}
\norm{d' - d} &= d'(m,m') - d(m,m') \\
&= \Tilde\Psi(d')(m,m') - \Tilde\Psi(d)(m,m') \tag{by $d = \Tilde{\Psi}(d)$ and $d' = \Tilde{\Psi}(d')$} \\
&= \Psi(d')(m,m') - \Psi(d)(m,m') \tag{by $m \not\bisim m'$ and def.\ $\Tilde\Psi$} \\
&\leq \textstyle \sum_{(x,y) \in E}(\Lambda(d')(x,y) - \Lambda(d)(x,y))\cdot \omega(x,y) \tag{as proved in Lemma~\ref{lem:PsiContinuous}} \,,
\end{align*}
where we recall that $E = \set{((a,n),(a,n'))}{ a\in \Labels, n,n' \in M}$.

Observe that $(\Lambda(d') - \Lambda(d))((a,n),(a,n')) = d'(n) - d(n') \leq \norm{d' - d}$, for all $n,n' \in M$ and $a \in \Labels$.
Since $\norm{d' - d} > 0$ the inequality 
$$\norm{d' - d} \leq \textstyle\sum_{(x,y) \in E}(\Lambda(d')(x,y) - \Lambda(d)(x,y))\cdot \omega(x,y) \leq \norm{d' - d}$$ 
holds only if the support of $\omega$ is included in $E_R = \set{((a,n),(a,n'))}{ a\in \Labels \text{ and } n \mathrel{R} n'}$. Since the argument holds for arbitrary $m,m' \in M$ such that $m \mathrel{R} m'$, we have that $R$ is a bisimulation, which is in contradiction with the initial assumptions.
\end{proof}

\subsection{A Quantitative Algebra of Open Markov Chains}

We turn the algebra of pointed open Markov chains $(\OMC, \O^\textsf{omc})$ 
given in Section~\ref{sec:AOMC} into a ``quantitative algebra'' by endowing it with the probabilistic 
bisimilarity pseudometric of Desharnais et al.
Our definition, however, does not comply with the non-expansivity conditions that are required 
for the interpretations of the algebraic operators by Definition~\ref{def:quantitativealgebra}. 
Indeed we show that the recursion operator fails to be non-expansive.

We conclude the section by extending the universality result of Theorem~\ref{th:universal} to the 
quantitative setting.

\medskip 
In Definition~\ref{def:bisimdist}, the bisimilarity pseudometric $\dist[\M]$ is defined over the states of 
a given open Markov chain $\M$. This can be extended to distance 
$\dist[\OMC] \colon \OMC \times \OMC \to [0,1]$ over the set $\OMC$ of open Markov 
chains by simply computing the bisimilarity distance between the initial states on the disjoint 
union of the two open Markov chains, \ie 
\begin{equation*}
  \dist[\OMC]((\M,m), (\N,n)) = \dist[\M \oplus \N](m,n) \,.
\end{equation*}

\begin{defi} \label{def:qaOMC}
The \emph{quantitative algebra of open Markov chains} is $(\OMC, \O^\textsf{omc}, \dist[\OMC])$.
\end{defi}
It is important to remark that the algebraic structure we just defined \emph{is not} a quantitative algebra
in the sense of~\cite{MardarePP:LICS16} (\cf~Definition~\ref{def:quantitativealgebra}), because as we
show in Example~\ref{ex:recNonexpansive} the interpretation of the recursion operator fails to be non-expansive.

\begin{rem} \label{rem:relaxedAlgebra}
The use of the term ``quantitative algebra'' in Defintion~\ref{def:qaOMC} is clearly an abuse of terminology. 
Perhaps, we should call such structures ``\emph{relaxed quantitative algebras}'' to emphasize the
fact that the interpretations of the operators do not need to be non-expansive. For the sake of 
readability, however, we decided to omit the adjective ``relaxed'' since, as it will be showed in
Sections~\ref{sec:soundness} and \ref{sec:completeness}, this detail will not cause troubles for 
the soundness and completeness results.
\end{rem}

\begin{exa}[Recursion is not non-expansive!] \label{ex:recNonexpansive}
We show an example where the recursion operator $\rec{X}{}$ fails to be non-expansive w.r.t.\ the
bisimilarity distance.

Let $0 < \varepsilon < \frac{1}{2}$ and consider the two pointed open Markov chains depicted below.
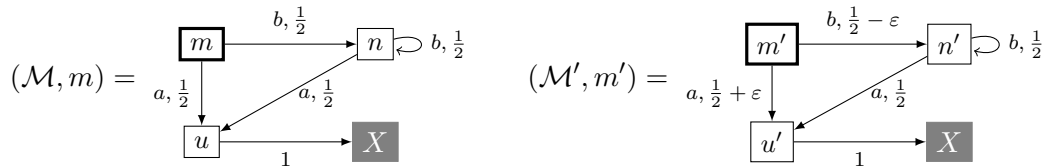
\def\skiph{2}
\def\skipv{1.3}
\begin{align*}
(\M,m) = 
\tikz[labels, baseline={(current bounding box.center)}]{ 
  \draw (0,0) node[state, initial] (m) {$m$};
  \draw ($(m)+(down:\skipv)$) node[state] (u) {$u$}; 
  \draw ($(m.east)+(right:\skiph)$) node[state] (n) {$n$};
  \draw ($(n)+(down:\skipv)$) node[vars] (x) {$X$};
  \path[-latex, font=\scriptsize]
    (m) edge node[above] {$b,\frac{1}{2}$} (n)
    (m) edge node[left] {$a,\frac{1}{2}$} (u)
    (n) edge[loop right] node[right] {$b,\frac{1}{2}$} (n)
    (n) edge node[right] {$a,\frac{1}{2}$} (u)
    (u) edge node[below] {$1$} (x);
}
&&
(\M',m') = 
\tikz[labels, baseline={(current bounding box.center)}]{ 
  \draw (0,0) node[state, initial] (m) {$m'$};
  \draw ($(m)+(down:\skipv)$) node[state] (u) {$u'$}; 
  \draw ($(m.east)+(right:\skiph)$) node[state] (n) {$n'$};
  \draw ($(n)+(down:\skipv)$) node[vars] (x) {$X$};
  \path[-latex, font=\scriptsize]
    (m) edge node[above] {$b,\frac{1}{2}-\varepsilon$} (n)
    (m) edge node[left] {$a,\frac{1}{2}+\varepsilon$} (u)
    (n) edge[loop right] node[right] {$b,\frac{1}{2}$} (n)
    (n) edge node[right] {$a,\frac{1}{2}$} (u)
    (u) edge node[below] {$1$} (x);
}
\end{align*}
By a straightforward computation one can readily show that the bisimilarity distance between 
$(\M,m)$ and $(\M',m')$ is $\dist((\M,m),(\M',m')) = \varepsilon$.
Let now consider the application of the operator $(\rec{X}{})^\textsf{omc}$ on these two pointed 
open Markov chains. It turns out that the resulting chains have the following 
behavior: 
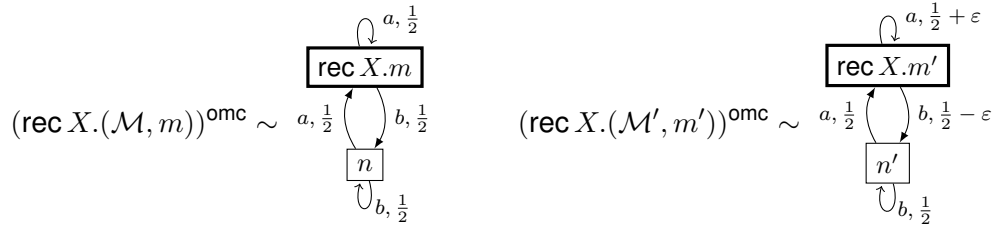
\begin{align*}
(\rec{X}{(\M,m)})^\textsf{omc} \bisim
\tikz[labels, baseline={([yshift=-1.5ex]current bounding box.center)}]{ 
  \draw (0,0) node[state, initial] (m) {$\rec{X}{m}$}; 
  \draw ($(m)+(down:\skipv)$) node[state] (n) {$n$};
  \path[-latex, font=\scriptsize]
    (m) edge[bend left] node[right] {$b,\frac{1}{2}$} (n)
    (m) edge[loop above] node[right,xshift=1mm] {$a,\frac{1}{2}$} (m)
    (n) edge[bend left] node[left] {$a,\frac{1}{2}$} (m)
    (n) edge[loop below] node[right] {$b,\frac{1}{2}$} (n);
}
&&
(\rec{X}{(\M',m')})^\textsf{omc} \bisim
\tikz[labels, baseline={([yshift=-1.5ex]current bounding box.center)}]{ 
  \draw (0,0) node[state, initial] (m) {$\rec{X}{m'}$}; 
  \draw ($(m)+(down:\skipv)$) node[state] (n) {$n'$};
  \path[-latex, font=\scriptsize]
    (m) edge[bend left] node[right] {$b,\frac{1}{2}-\varepsilon$} (n)
    (m) edge[loop above] node[right,xshift=1mm] {$a,\frac{1}{2}+\varepsilon$} (m)
    (n) edge[bend left] node[left] {$a,\frac{1}{2}$} (m)
    (n) edge[loop below] node[right] {$b,\frac{1}{2}$} (n);
}
\end{align*}
By an easy calculation, independently of the value $\varepsilon \in (0,\frac{1}{2})$, we obtain that
\begin{equation*}
  \dist((\rec{X}{(\M,m)})^\textsf{omc}, (\rec{X}{(\M',m')})^\textsf{omc}) = 1 \,.
\end{equation*}
Since $\varepsilon < 1$, $\dist((\M,m),(\M',m')) \not\geq \dist((\rec{X}{(\M,m)})^\textsf{omc}, (\rec{X}{(\M',m')})^\textsf{omc})$.
This proves that 
$(\rec{X}{})^\textsf{omc}$ fails to be non-expansive w.r.t.\ the bisimilarity distance.
\qed
\end{exa}

By Theorem~\ref{th:universal}, we know that one can equivalently reason about bisimilarity  
between the semantics of terms by simply considering bisimilarity for the corresponding terms as states in the universal 
open Markov chain $\UU$. The next result states that the situation is similar when one needs to 
compute the distance between the semantics of terms. 
\begin{thm}[Quantitative universality] \label{th:universaldist}
Let $t,s \in \TT{}$. Then $\dist[\OMC](\denot{t},\denot{s}) = \dist[\UU](t,s)$.
\end{thm}
\begin{proof}
The equality $\dist[\OMC](\denot{t},\denot{s}) = \dist[\UU](t,s)$ follows by Lemma~\ref{lem:bisimdist} 
and Theorem~\ref{th:universal}.
\begin{align*}
 \dist[\OMC](\denot{t},\denot{s})
 &= \dist(\denot{t},\denot{s}) \tag{def. $\dist$} \\ 
 &\leq \dist(\denot{t}, (\UU,t)) + \dist((\UU,t),(\UU,s)) + \dist((\UU, s),\denot{s}) \tag{triangular ineq.} \\
 &= \dist((\UU,t),(\UU,s)) \tag{Theorem~\ref{th:universal} \& Lemma~\ref{lem:bisimdist}} \\
 &= \dist[\UU](t,s) \,. \tag{def. $\dist$}
\end{align*}
By a similar argument we also have $\dist[\OMC](\denot{t},\denot{s}) \geq \dist[\UU](t,s)$, hence the thesis.
\end{proof}

\subsection{A Quantitative Deduction System}
\label{sec:quantAxioms}
Now we present a quantitative deduction system which will be later shown to be sound and complete
w.r.t.\ the probabilistic bisimilarity distance of Desharnais et al.. 
The deduction system we propose will not be a quantitative deduction system in the sense of~\cite{MardarePP:LICS16}, because it does not satisfies the (\Nexp) axiom of non-expansivity 
for the signature operators (\cf\ Section~\ref{sec:QDS}).  

\medskip
The quantitative deduction system ${\vdash} \subseteq 2^{\E[\Sigma]} \times \E[\Sigma]$ of type $\O$ 
that we consider satisfies the axioms (\Refl), (\Symm), (\Triang), (\Max), (\Arch) and rules (\Subst), (\Cut) (\Assum) from Section~\ref{sec:QDS} and the following additional axioms
\begin{align*} 
(\Bone)\;\; 
& \vdash t +_1 s \equiv_0 t \,, \\
(\Btwo)\;\; 
& \vdash t +_e t \equiv_0 t \,, \\
(\SC)\;\; 
& \vdash t +_e s \equiv_0 s +_{1-e} t \,, \\
(\SA)\;\; 
& \vdash (t +_e s) +_{e'} u \equiv_0 t +_{ee'} (s +_{\frac{e' - ee'}{1 - ee'}} u) \,, \text{ for $e,e' \in [0,1)$} \,, \\ 
(\Unfold)\;\; & \vdash \rec{X}{t} \equiv_0 t[\rec{X}{t} / X] \,, \\
(\Unguard)\;\; & \vdash \rec{X}{(t +_e X)} \equiv_0 \rec{X}{t} \,, \\
(\Fix)\;\; & \{ s \equiv_0 t[s / X] \} \vdash s \equiv_0 \rec{X}{t} \,, \text{ for $X$ guarded in $t$} \,, \\
(\Cong)\;\; & \{ t \equiv_0 s \} \vdash \rec{X}{t} \equiv_0 \rec{X}{s} \,,  \\[1ex]
(\Top)\;\; & \vdash t \equiv_1 s \,, \\
(\Pref)\;\; & \{ t \equiv_\e s \} \vdash \pref{a}{t} \equiv_{\e} \pref{a}{s} \,, \\
(\Conv)\;\; 
& \{ t \equiv_\e s, t' \equiv_{\e'} s' \} \vdash t +_e t' \equiv_{\e''} s +_e s' \,, \text{ for $\e'' \geq e \e + (1-e) \e'$} \,. 
\end{align*}
Note that the axiom (\Nexp) is not included in the definition.

\begin{rem} \label{rem:relaxedDedSys}
The use of the term ``quantitative deduction system'' is again an abuse of terminology. To emphasize the
fact that the axiom (\Nexp) is not required to be satisfied, we should perhaps have used 
the word ``\emph{relaxed quantitative deduction system}'' to mark the difference w.r.t.\ the quantitative
algebraic framework of Mardare, Panangaden, and Plotkin~\cite{MardarePP:LICS16}. 
However the omission of the adjective ``relaxed'' will not cause troubles in the further development of 
the paper.
\end{rem}

(\Bone), (\Btwo), (\SC), (\SA) are the axioms of \emph{barycentric algebras} (a.k.a.\ 
\emph{abstract convex sets}) due to M.\ H.\ Stone~\cite{Stone49}, used here to axiomatize the convex set of probability distributions. (\SC) stands for \emph{skew commutativity} and (\SA) for \emph{skew associativity}.
Barycentric algebras are \emph{entropic} in the sense that all operations $+_e$ are affine maps, that is, 
for all $e,d \in [0,1]$ we have the entropic identity
\begin{equation*}
  \vdash (t +_e s) +_d (t' +_e s') \equiv_0 (t +_d t') +_e (s +_d s') \,.
  \tag{\Entr}
\end{equation*}
If $t = s$, by (\Btwo) the above reduces to the distributivity law 
\begin{equation*}
  \vdash u +_d (t' +_e s')  \equiv_0 (u +_d t') +_e (u +_d s') \,.
  \tag{\Distr}
\end{equation*}
Although the entropic identity (\Entr) can be verified by direct deduction, a simpler proof for it that 
does not use a direct approach can be found in~\cite[Lemma 2.3]{KeimelP15}.

The axioms (\Unfold), (\Unguard), (\Fix), (\Cong) are the recursion axioms of Milner~\cite{Milner84}, used 
here to axiomatize the coinductive behavior of open Markov chains. (\Unfold) and (\Fix) state that, whenever $X$ is guarded in a term $t$, $\rec{X}{t}$ is \emph{the unique solution} of the recursive 
equation $s \equiv_0 t[s / X]$.
The axiom (\Unguard) deals with unguarded recursive behavior, and (\Cong) states the 
congruential properties of the recursion operator.

As opposed to the axioms described so far, which are essentially equational, the last three are
the only truly quantitative one characterizing the quantitative deduction system we have just introduced. 
(\Top) states that the distance between terms is bounded by $1$; (\Pref) is the non-expansivity for the 
prefix operator; and (\Conv) is the \emph{interpolative barycentric axiom} of Mardare, Panangaden, 
and Plotkin introduced in~\cite{MardarePP:LICS16} for axiomatizing the Kantorovich 
distance between finitely-supported probability distributions (\cf\ \S10 in~\cite{MardarePP:LICS16}).

Note that for $\e = \e'$, (\Conv) reduces to non-expansivity for the operator $+_e$:
\begin{equation*}
  \{ t \equiv_\e s, t' \equiv_{\e} s' \} \vdash t +_e t' \equiv_{\e} s +_e s' \,,
\end{equation*}
and non-expansivity always entails congruence for the operator. 
Indeed, for $\e = \e' = 0$ in (\Pref) and (\Conv) we obtain 
\begin{align*}
  \{ t \equiv_0 s \} &\vdash \pref{a}{t} \equiv_{0} \pref{a}{s} \,, 
  \tag{\Pref-0}
  \\
  \{ t \equiv_0 s, t' \equiv_0 s' \} &\vdash t +_e t' \equiv_0 s +_e s' \,,
  \tag{\Conv-0}
\end{align*}
corresponding to the congruence for the prefix and probabilistic choice operators, respectively.

It is important to remark that the quantitative deduction system presented above subsumes
the equational deduction system of Stark and Smolka~\cite{StarkS00} that axiomatizes probabilistic
bisimilarity. 

We conclude this section by recalling some historical notes from~\cite{KeimelP15,CIS-300992} 
about the axioms of barycentric algebras and recursion.
\begin{rem}[Historical note]
The first axiomatization of convex sets can be traced back to M.\ H.\ Stone~\cite{Stone49}. 
Independently, Kneser~\cite{Kneser52} gave a similar axiomatization. Stone's and Kneser's axioms where
not restricted to convex sets arising in vector spaces over the reals but, by requiring an additional 
cancellation axiom, they axiomatized convex sets 
\emph{embeddable} in vector spaces over linearly ordered skewed fields.
W.\ D.\ Neumann~\cite{Neumann1970} seems to be the first to have looked at a truly equational theory of convex sets.
He remarked that barycentric algebras may be very different from convex sets in vector spaces. 
Indeed $\vee$-semilattices are an example of barycentric algebra by interpreting $+_e$ as $\vee$,
for all $0 < e < 1$, and $+_1$ and $+_0$ as left and right projections, respectively.

The axioms (\Bone), (\Btwo), (\SC), (\SA) that we use in this work are due to \v{S}wirszcz~\cite{Swirszcz}
and have been reproduced by Romanowska and Smith~\cite{modes}, who actually introduced the terminology
\emph{barycentric algebra} for an abstract convex set.

Early attempts to prove equational properties of recursive definitions in various specific contexts
include the work of de Bakker~\cite{deBakker:recursive}, Manna and Vuillemin~\cite{MannaV72}, and
Kahn~\cite{Kahn74}. Since then, the general study of recursion equations has been pursued 
under several guises: as recursive applicative program schemes~\cite{Courcelle1974}, 
$\mu$-calculus~\cite{BloomE94}, and perhaps 
most notably as the iteration theories~\cite{BloomE93} of Bloom and \'Esik.

The axioms (\Unfold), (\Unguard), (\Fix), (\Cong) first appeared in Milner~\cite{Milner84}, who
was certainly aware of de Bakker's work (\cf~\cite{milner1975}). The same axioms have
been used by Stark and Smolka in their development of an equational deduction 
system for probabilistic bisimilarity~\cite{StarkS00}. An equivalent axiomatization to that 
in~\cite{StarkS00} appeared in~\cite{AcetoEI02}, where recursion is extended to vector terms
and the recursive axioms were replaced by the axioms of iteration algebras, 
a.k.a.\ Conway equations~\cite{BloomE93},
\begin{align*}
  \vdash \rec{X}{t[s/X]} &\equiv_0 t[\rec{X}{s[t/X]}/X] 
  \tag{Composition Identity} \\
  \vdash \rec{X}{t[X/Y]} &\equiv_0 \rec{X}{\rec{Y}{t}}
  \tag{Diagonal Identity}
\end{align*}
capturing the equational properties of the fixed point operations in a purely equational way. 
Note that the composition identity reduces to (\Unfold) when taking $s = X$.
\end{rem}

\subsection{Soundness} 
\label{sec:soundness}

In this section we show the soundness of our quantitative deduction system w.r.t.\ the bisimilarity distance
between pointed open Markov chains.

\medskip
Recall that, by Theorem~\ref{th:universaldist}, it is totally equivalent to reason about the distance 
between $\denot{t}$ and $\denot{s}$ by just considering the bisimilarity distance between the states 
corresponding to the terms $t$ and $s$ in the universal open Markov chain $\UU$. 
Hence hereafter, whenever we refer to 
the distance between terms in $\TT{}$ we will use $\dist[\UU]$, often simply denoted as $\dist$.
Similarly, $\models_\OMC t \equiv_\e s$ is equivalent to $\models_\UU t \equiv_\e s$, and it 
will be denoted just by $\models t \equiv_\e s$.

\begin{thm}[Soundness] \label{th:soundness}
For arbitrary $t, s \in \TT{}$, if $\vdash t \equiv_\e s$ then $\models t \equiv_\e s$.
\end{thm}
\begin{proof}
We must show that each axiom and deduction rule of inference is valid. The axioms (\Refl), (\Symm), (\Triang), (\Max), and (\Arch) are sound since $\dist$ is a pseudometric (Lemma~\ref{lem:bisimdist}). The soundness of the classical logical deduction rules (\Subst), (\Cut), and (\Assum) is immediate.
By Lemma~\ref{lem:bisimdist}, the kernel of $\dist$ is $\bisim$. Hence the axioms of barycentric algebras (\Bone), (\Btwo), (\SC), and (\SA) along with (\Unfold), (\Unguard), (\Cong), and (\Fix) follow directly by the soundness theorem proven in~\cite{StarkS00} (\cf\ Remark~\ref{rmk:souness&completenessStarkSmolka}).

The soundness for the axiom (\Top) is immediate consequence of the fact that $\dist$ is $1$-bounded.
To prove the soundness of (\Pref) it suffices to show that $\dist(t,s) \geq \dist(\pref{a}{t}, \pref{a}{s})$.
\begin{align*}
  \dist(\pref{a}{t}, \pref{a}{s})
  &= \K[\Lambda(\dist)]{\tauTT^*(\pref{a}{t}), \tauTT^*(\pref{a}{s})} \tag{$\dist$ fixed-point \& def.\ $\Psi$} \\
  &= \K[\Lambda(\dist)]{\ind{\{(a,t)\}}, \ind{\{(a,s)\}}} \tag{def.\ $\tauTT$ \& $\PP_\UU$} \\
  &= \Lambda(\dist)((a,t), (a,s)) \tag{def.\ $\K[]{}$} \\
  &= \dist(t,s) \,. \tag{def.\ $\Lambda$}
\end{align*}
The soundness of (\Conv) follows by proving $e \dist(t,s) + (1-e) \dist(t',s') \geq \dist(t +_e t', s +_e s')$.
\begin{align*}
  e \dist(t,s) &+ (1-e) \dist(t',s') \\
  &= e \Psi(\dist)(t,s) + (1-e) \Psi(\dist)(t',s')  \tag{$\dist$ fixed point} \\
  &= e \K[\Lambda(\dist)]{\tauTT^*(t),\tauTT^*(s)} + (1-e) \K[\Lambda(\dist)]{\tauTT^*(t'),\tauTT^*(s')}  \tag{def.\ $\Psi$}
\intertext{then, for $\omega \in \coupling{\tauTT^*(t)}{\tauTT^*(s)}$ and $\omega' \in \coupling{\tauTT^*(t')}{\tauTT^*(s')}$ optimal couplings for $\K[\Lambda(\dist)]{}$, and by noticing that $e \omega + (1-e) \omega' \in \coupling{e\tauTT^*(t)+(1-e)\tauTT^*(t')}{e\tauTT^*(s)+(1-e)\tauTT^*(s')}$ we have}
  &= \textstyle e \sum_{x,y} \Lambda(\dist)(x,y) \cdot \omega(x,y) + (1-e) \sum_{x,y} \Lambda(\dist)(x,y) \cdot \omega'(x,y) \\
  &= \textstyle \sum_{x,y} \Lambda(\dist)(x,y) \cdot (e \cdot \omega(x,y) + (1-e) \cdot \omega'(x,y)) \tag{linearity} \\
  &\geq \K[\Lambda(\dist)]{e \tauTT^*(t) + (1-e) \tauTT^*(t'),e \tauTT^*(s) + (1-e) \tauTT^*(s')} 
    \tag{def.\ $\K[]{}$ and above} \\
  &= \K[\Lambda(\dist)]{\PP_\UU(\tauTT)^*(t +_e t'),\PP_\UU(\tauTT)^*(s +_e s')} \tag{def.\ $\PP_\UU$} \\
  &= \K[\Lambda(\dist)]{\tauTT^*(t +_e t'),\tauTT^*(s +_e s')} \tag{$\tauTT$ fixed-point} \\
  &= \dist(t +_e t', s +_e s') \tag{def.\ $\Psi$ \& $\dist$ fixed-point}
\end{align*}
The above concludes the proof.
\end{proof}

\subsection{Completeness} 
\label{sec:completeness}

This section is devoted to proving the completeness of our quantitative deduction system 
w.r.t.\ the bisimilarity distance between pointed open Markov chains.

\medskip
For the sake of readability it will be convenient to introduce the following notation for \emph{formal sums 
of terms} (or \emph{convex combinations of terms}). For $n \geq 1$, $t_1,\ldots, t_n \in \TT{}$ 
terms, and $e_1, \dots, e_n \in [0,1]$ positive reals such that $\sum_{i = 1}^n e_i = 1$, we define
\begin{equation*}
  \sum_{i=1}^n e_i \cdot t_i =
  \begin{cases}
    t_1 & \text{if $e_1 = 1$} \\
    t_1 +_{e_1} \left( \sum_{i=2}^n \frac{e_i}{1-e_1} \cdot t_i \right) & \text{otherwise} \,.
  \end{cases}
\end{equation*}

\medskip
Following the pattern of~\cite{Milner84,StarkS00}, the completeness theorem hinges on a couple of important transformations.
The first of these is the standard de Beki\v{c}-Scott construction of solutions of simultaneous recursive definitions. References for this theorem may be found in de Bakker~\cite{deBakker:recursive}, 
who seems the first to use it to support a proof rule for ``program equivalence''.
This is embodied in the next theorem, which is~\cite[Theorem 5.7]{Milner84}.
\begin{thm}[Unique Solution of Equations] \label{th:uniquesolution}
Let $\ol{X} = (X_1,\dots,X_k)$ and $\ol{Y} = (Y_1,\dots,Y_h)$ be distinct names, and $\ol{t} = (t_1,\dots,t_k)$ terms with free names in $(\ol{X},\ol{Y})$ in which each $X_i$ is guarded. Then there exist terms $\ol{s} = (s_1,\dots,s_k)$ with free names in $\ol{Y}$ such that 
\begin{align*}
&\vdash s_i \equiv_0 t[\ol{s} / \ol{X}] \,, && \text{for all $i \leq k$.}
\end{align*}
Moreover, if for some terms $\ol{u} = (u_1,\dots,u_k)$ with free variables in $\ol{Y}$, $\vdash u_i \equiv_0 t[\ol{u} / \ol{X}]$, for all $i \leq k$, then $\vdash s_i \equiv_0 u_i$, for all $i \leq k$.
\end{thm}

The second transformation provides a deducible normal form for terms. This
result is embodied in the following theorem, which is \cite[Theorem 5.9]{StarkS00}%
\footnote{The formulation given here is slightly simpler than the original one in~\cite{StarkS00}, since 
our deduction system satisfies the axiom (\Bone), which is not included in the equational deduction
system of~\cite{StarkS00}.}.
\begin{thm}[Equational Characterization] \label{th:eqcharacterization}
For any term $t$, with free names in $\ol{Y}$, there exist terms $t_1, \ldots, t_k$ with free names in $\ol{Y}$, such that
$\vdash t \equiv_0 t_1$ and 
\begin{align*}
  \vdash t_i \equiv_0 \sum_{j=1}^{h(i)} p_{ij} \cdot s_{ij} + \sum_{j=1}^{l(i)} q_{ij} \cdot Y_{g(i,j)} \,,
  && \text{for all $i \leq k$} \,,
\end{align*}
where the terms $s_{ij}$ and names $Y_{g(i,j)}$ are enumerated without repetitions, and $s_{ij}$ is either $\rec{X}{X}$ or has the form $\pref{a_{ij}}{t_{f(i,j)}}$.
\end{thm}

The last lemma relates the proposed deduction system with the Kantorovich distance
between probability distributions. So far this is the only transformation embodying the use of the interpolative
barycentric axiom (\Conv) to deduce quantitative information on terms.
\begin{lem} \label{lem:deduceKantorovich}
Let $d$ be a $1$-bounded pseudometric over $\TT{}$ and $\mu,\nu \in \Delta(\TT{})$ probability measures
with supports $\mathit{supp}(\mu) = \{ t_1, \dots, t_k \}$ and $\mathit{supp}(\nu) = \{ s_1, \dots, s_r \}$. 
Then 
\begin{align*}
  \set{t_i \equiv_\e s_u}{\e \geq d(t_i,s_u)}
  \vdash
  \sum_{i = 1}^k \mu(t_i) \cdot t_i \equiv_{\e'} \sum_{u = 1}^r \nu(s_u) \cdot s_u \,,
  &&
  \text{for all $\e' \geq \K[d]{\mu,\nu}$} \,.
\end{align*}
\end{lem}
\begin{proof}
We proceed by well-founded induction on the strict preorder 
\begin{align*}
(\mu,\nu) \prec (\mu',\nu') && \text{iff} &&
\left\{\begin{array}{c} 
\text{($\mathit{supp}(\mu) \subset \mathit{supp}(\mu')$ 
and $\mathit{supp}(\nu) \subseteq \mathit{supp}(\nu')$)} \,\phantom{.}\\
\text{or} \\
\text{($\mathit{supp}(\mu) \subseteq \mathit{supp}(\mu')$ 
and $\mathit{supp}(\nu) \subset \mathit{supp}(\nu')$)} \,.
\end{array}\right.
\end{align*}

(Base case: $\mathit{supp}(\mu) = \{t_1\}$ and $\mathit{supp}(\mu') = \{ s_1 \}$). 
In this case $\mu = \ind{\{t_1\}}$ and $\nu = \ind{\{s_1\}}$. Thus the proof
follows by (\Assum), by noticing that $\K[d]{\mu, \nu} = d(t_1,s_1)$.

(Inductive step: $\mathit{supp}(\mu) = \{ t_1, \dots, t_k \}$ and $\mathit{supp}(\nu) = \{ s_1, \dots, s_r \}$) Assume without loss of generality that $k > 1$ (if $k = 1$, then $r > 1$ and we proceed
dually). The proof is structured as follows. We find suitable $e \in (0,1)$ and $\mu_1,\mu_2,\nu_1, \nu_2 \in \Delta(\TT{})$, such that 
\begin{enumerate}
  \item \label{step1} $(\mu_1,\nu_1) \prec (\mu,\nu)$ and $(\mu_2,\nu_2) \prec (\mu,\nu)$; 
  \item \label{step2} $\K[d]{\mu,\nu} = e \K[d]{\mu_1,\nu_1} + (1-e) \K[d]{\mu_2,\nu_2}$;
  \item \label{step3} and the following are deducible
  \begin{align}
  \textstyle
  \vdash \sum_{i = 1}^k \mu(t_i) \cdot t_i &\equiv_0 
  \textstyle
  	\big( \sum_{i = 1}^k \mu_1(t_i) \cdot t_i \big) +_e \big( \sum_{i = 1}^k \mu_2(t_i) \cdot t_i \big) \,,
  \label{eq:decuceSplitting1}
  \\
  \textstyle
  \vdash \sum_{u = 1}^r \nu(s_u) \cdot s_u &\equiv_0
  \textstyle
  	\big( \sum_{u = 1}^r \nu_1(s_u) \cdot s_u \big) +_e \big( \sum_{u = 1}^r \nu_2(s_u) \cdot s_u \big) \,.
  \label{eq:decuceSplitting2}
  \end{align}
\end{enumerate}
By \eqref{step1} and the inductive hypothesis, we have that, for $j \in \{1,2 \}$
\begin{align*}
  \set{t_i \equiv_\e s_u}{\e \geq d(t_i,s_u)}
  \vdash
  \sum_{i = 1}^k \mu_j(t_i) \cdot t_i \equiv_{\e'} \sum_{u = 1}^r \nu_j(s_u) \cdot s_u \,,
  &&
  \text{for all $\e' \geq \K[d]{\mu_j,\nu_j}$} \,.
\end{align*}
From the above, \eqref{step3}, and (\Conv) we deduce
\begin{align*}
  \set{t_i \equiv_\e s_u}{\e \geq d(t_i,s_u)}
  \vdash
  \sum_{i = 1}^k \mu(t_i) \cdot t_i \equiv_{\e'} \sum_{u = 1}^r \nu(s_u) \cdot s_u \,,
  &&
  \text{for all $\e' \geq \kappa$} \,.
\end{align*}
where $\kappa = e \K[d]{\mu_1,\nu_1} + (1-e) \K[d]{\mu_2,\nu_2}$. Then, the proof follows from \eqref{step2}.

In the following we provide the definitions for $e \in (0,1)$ and  $\mu_1,\mu_2,\nu_1, \nu_2 \in \Delta(\TT{})$, then in turn we prove \eqref{step1}, \eqref{step2}, and \eqref{step3}. Let $e = \mu(t_1)$. Note that
since $\mathit{supp}(\mu) = \{ t_1, \dots, t_k\}$ and $k > 1$, we have that $\mu(t_1) \in (0,1)$.
Let $\tilde\omega \in \coupling{\mu}{\nu}$ be the minimal coupling for $\K[d]{\mu,\nu}$, \ie, the one realizing 
the following equality 
(\cf\ the definition of $\K[d]{}$)
\begin{equation}
\textstyle
  \K[d]{\mu,\nu} = \sum_{i \leq k,\, u \leq r} d(t_i,s_u) \cdot \tilde\omega(t_i,s_u) \,,
  \label{eq:optimalcoupling}
\end{equation} 
and, for $2 \leq i \leq k$, $1 \leq u \leq r$, define
\begin{align*}
 \mu_1(t_1) = 1 \,,
 &&
 \mu_2(t_i) = \frac{\mu(t_i)}{1-\mu(t_1)} \,,
 &&
 \nu_1(s_u) = \frac{\tilde\omega(t_1,s_u)}{\mu(t_1)} \,,
 && 
 \nu_2(s_u) = \frac{\nu(s_u) - \tilde\omega(t_1,s_u)}{1-\mu(t_1)} \,.
\end{align*}
Note that $\mathit{supp}(\mu_1) = \{t_1\}$, $\mathit{supp}(\mu_2) = \{t_2, \dots, t_k\}$ and
$\mathit{supp}(\nu_1), \mathit{supp}(\nu_2) \subseteq \mathit{supp}(\nu)$. 
It is easy to show that, since $\tilde\omega \in \coupling{\mu}{\nu}$, the above are 
well-defined probability distributions.

\eqref{step1} It follows directly by definition of $\prec$, $\mathit{supp}(\mu_1) = \{t_1\}$, 
$\mathit{supp}(\mu_2) = \{t_2, \dots, t_k\}$, $\mathit{supp}(\mu) = \{ t_1, \dots, t_k \}$, and
$\mathit{supp}(\nu_1), \mathit{supp}(\nu_2) \subseteq \mathit{supp}(\nu)$. 

\eqref{step2}
Define the measures $\tilde\omega_1, \tilde\omega_2$ with
supports $ \mathit{supp}(\tilde\omega_1) = \set{(t_1,s_u)}{1 \leq u \leq r}$ and $\mathit{supp}(\tilde\omega_2) = \set{(t_i,s_u)}{2 \leq i \leq k,\, 1 \leq u \leq r}$ as follows
\begin{align*}
  \tilde\omega_1(t_1,s_u) = \frac{\tilde\omega(t_1,s_u)}{\mu(t_1)} \,;
  &&
  \tilde\omega_2(t_i,s_u) = \frac{\tilde\omega(t_i,s_u)}{1 - \mu(t_1)}\,, 
  &&
  \text{for $2 \leq i \leq k$ and $1 \leq u \leq r$.} 
\end{align*}
By the fact that $\tilde\omega \in \coupling{\mu}{\nu}$, one easily get 
that $\tilde\omega_1 \in \coupling{\mu_1}{\nu_1}$ and $\tilde\omega_2 \in \coupling{\mu_2}{\nu_2}$. From this, the following inequality holds:
\begin{align*}
 \K[d]{\mu,\nu}
 &= \sum_{i,\, u} d(t_i,s_u) \cdot \tilde\omega(t_i,s_u) 
 \tag{by Equation~\ref{eq:optimalcoupling}}
 \\
 &= \mu(t_1) \left( \sum_{i,\, u} d(t_i,s_u) \cdot \tilde\omega_1(t_i,s_u) \right) +
      (1-\mu(t_1) ) \left( \sum_{i,\, u} d(t_i,s_u) \cdot \tilde\omega_2(t_i,s_u) \right) \\
 &\geq \mu(t_1) \K[d]{\mu_1,\nu_1} + (1-\mu(t_1)) \K[d]{\mu_2,\nu_2} \,.
\end{align*}

Now, we prove also that the reverse inequality holds. 
Assume that, for $j \in \{1,2\}$, $\tilde\omega_j$ is the minimal 
coupling for $\K[d]{\mu_j,\nu_j}$. By the fact that $\mathit{supp}(\mu_1) = \{t_1\}$, 
$\mathit{supp}(\mu_2) = \{t_2, \dots, t_k\}$ is a partition of $\mathit{supp}(\mu)$, we can 
define the coupling $\omega \in \coupling{\mu}{\nu}$ as follows, for $2 \leq i \leq k$ and $1 \leq u \leq r$:
\begin{align*}
  \omega(t_1,s_u) = \mu(t_1) \cdot \tilde\omega_1(t_1,s_u) \,;
  &&
  \omega(t_i,s_u) = (1 - \mu(t_1)) \cdot \tilde\omega_2(t_i,s_u) \,.
\end{align*}
For this, the following inequality holds:
\begin{align*}
\mu(t_1) &\K[d]{\mu_1,\nu_1} + (1-\mu(t_1)) \K[d]{\mu_2,\nu_2} \\
&= \mu(t_1) \left( \sum_{i,\, u} d(t_i,s_u) \cdot \tilde\omega_1(t_i,s_u) \right) +
      (1-\mu(t_1) ) \left( \sum_{i,\, u} d(t_i,s_u) \cdot \tilde\omega_2(t_i,s_u) \right) \\
&= \sum_{i,\, u} d(t_i,s_u) \cdot \omega(t_i,s_u) \\
&\geq \K[d]{\mu,\nu} \,.
\end{align*}
Therefore, \eqref{step2} holds.

\eqref{step3} 
We start by showing \eqref{eq:decuceSplitting1}. Since $\mu(t_1) \in (0,1)$, the formal sum 
on the left-hand side of \eqref{eq:decuceSplitting1} is syntactically equivalent to
\begin{equation}
  \sum_{i = 1}^k \mu(t_i) \cdot t_i = 
  t_1 +_{\mu(t_1)} \left( \sum_{i = 2}^k \frac{\mu(t_i)}{1-\mu(t_1)} \cdot t_i \right) \,,
  \label{eq:sumMu}
\end{equation}
By (\Bone), (\SC), and the definitions of $\mu_1, \mu_2$ we easily obtain
\begin{align*}
  \vdash t_1 \equiv_0 \sum_{i = 1}^k \mu_1(t_i) \cdot t_i \,,
  &&
  \vdash \sum_{i = 2}^k \frac{\mu(t_i)}{1-\mu(t_1)} \cdot t_i \equiv_0 
  	\sum_{i = 1}^k \mu_2(t_i) \cdot t_i  \,.
\end{align*}
Thus \eqref{eq:decuceSplitting1} follows from the deductions above by applying (\Conv-0) to
\eqref{eq:sumMu}.
Next we prove \eqref{eq:decuceSplitting2} by showing that for any coupling $\omega \in \coupling{\mu}{\nu}$ the following is deducible:
\begin{equation}
  \vdash \sum_{u = 1}^r \nu(s_u) \cdot s_u 
  \equiv_0
  \left( \sum_{u = 1}^r \frac{\omega(t_1,s_u)}{\mu(t_1)} \cdot s_u \right)
  +_{\mu(t_1)}
  \left( \sum_{u = 1}^r \frac{\nu(s_u) - \omega(t_1,s_u)}{1-\mu(t_1)} \cdot s_u \right) \,.
  \label{eq:splitting}
\end{equation}
We do this by induction on the size of the support of $\nu$. 
(Base case: $\mathit{supp}(\nu) = \{s_1\}$). Then, $\nu(s_1) = 1$ and $\omega(t_1,s_1) = \mu(t_1)$, 
so \eqref{eq:splitting} reduces to (\Btwo).
(Inductive step: $r > 1$ and $\mathit{supp}(\nu) = \{s_1,\dots,s_r\}$). Then $\nu(s_1) \in (0,1)$. 
Thus, the formal sum on the left-hand side of \eqref{eq:splitting} is syntactically
equivalent to
\begin{equation}
  \sum_{u = 1}^r \nu(s_u) \cdot s_u = 
  s_1 +_{\nu(s_1)} \left( \sum_{u = 2}^r \nu'(s_u) \cdot s_u \right) \,,
  \label{eq:firstSlice}
\end{equation}
where $\nu'(s_u) = \frac{\nu(s_u)}{1-\nu(s_1)}$, for $2 \leq u \leq r$.
Note that $\omega'(t_i,s_u) = \frac{\omega(t_i,s_u)}{1-\nu(s_1)}$, for $1\leq i \leq k$ and $2 \leq u \leq r$, is
a coupling in $\coupling{\mu}{\nu'}$ and that $\mathit{supp}(\nu') = \{s_2, \dots, s_r \}$. Thus, by inductive hypothesis on $\nu'$ we obtain
\begin{align*}
  \vdash \sum_{u = 2}^r \nu'(s_u) \cdot s_u 
  &\equiv_0
  \left( \sum_{u = 2}^r \frac{\omega'(t_1,s_u)}{\mu(t_1)} \cdot s_u \right)
  +_{\mu(t_1)}
  \left( \sum_{u = 2}^r \frac{\nu'(s_u) - \omega'(t_1,s_u)}{1-\mu(t_1)} \cdot s_u \right)
  \\
  &=
  \left( \sum_{u = 2}^r \frac{\omega(t_1,s_u)}{\mu(t_1)(1- \nu(s_1))} \cdot s_u \right)
  +_{\mu(t_1)}
  \left( \sum_{u = 2}^r \frac{\nu(s_u) - \omega(t_1,s_u)}{(1-\mu(t_1))(1-\nu(s_1))} \cdot s_u \right)
\end{align*}
From this deduction and \eqref{eq:firstSlice}, by (\Dist), we obtain \eqref{eq:splitting}.
\end{proof}

Now we are ready to prove the main result of this section. The proof of completeness can be roughly 
sketched as follows. Given $t,s \in \TT{}$ such that $\dist(t,s) \leq \e$, to prove $\vdash t \equiv_\e s$
we first apply Theorem~\ref{th:eqcharacterization} to get their deducible equational normal forms as
formal sums. Then, for each $\alpha \in \naturals$, by applying (\Top) for the case $\alpha = 0$, and 
Lemma~\ref{lem:deduceKantorovich} and (\Pref) for $\alpha > 0$,
we deduce $\vdash t \equiv_\e s$, for all $\e \geq \Tilde\Psi^{\alpha}(\mathbf{1})(t,s)$.
Then, my (\Max) and (\Arch), $\vdash t \equiv_\e s$ follows by 
noticing that $\dist(t,s) = \bigsqcap_{\alpha \in \naturals} \Tilde{\Psi}^\alpha(\mathbf{1})(t,s)$
(Lemma~\ref{lem:uniquefix}).

\begin{thm}[Completeness]\label{th:completeness}
For arbitrary $t,s \in \TT{}$, if $\models t \equiv_\e s$, then $\vdash t \equiv_\e s$.
\end{thm}
\begin{proof}
Let $t,s \in \TT{}$ and $\e \in \Q$. We have to show that if $\dist(t,s) \leq \e$ then $\vdash t \equiv_\e s$.
The case $\e \geq 1$ trivially follows by (\Top) and (\Max). Let $\e < 1$.
By Theorem~\ref{th:eqcharacterization}, there exist terms $t_1, \ldots, t_k$ and $s_1, \ldots, s_r$ with free 
names in $\ol{X}$ and $\ol{Y}$, respectively, such that $\vdash t \equiv_0 t_1$, $\vdash s \equiv_0 s_1$, and   
\begin{align}
  &\textstyle
  \vdash t_i \equiv_0 \sum_{j=1}^{h(i)} p_{ij} \cdot t'_{ij} \,,
  & \text{for all $i \leq k$,} \label{eq:tterms} \\
  &\textstyle
  \vdash s_u \equiv_0 \sum_{v=1}^{n(u)} q_{uv} \cdot s'_{uv} \,,
  & \text{for all $u \leq r$,} \label{eq:sterms}
\end{align}
where the terms $t'_{ij}$ (resp.\ $s'_{uv}$) are enumerated without repetitions, and $t'_{ij}$ (resp.\ $s'_{uv}$) have either the form $\pref{a_{ij}}{t_{f(i,j)}}$ (resp.\ $\pref{b_{uv}}{s_{z(u,v)}}$), or $X_{g(i,j)}$ 
(resp.\ $Y_{w(u,v)}$), or $\rec{Z}{Z}$.

If we can prove that for all $\alpha \in \naturals$,
\begin{align}
  \vdash t_i \equiv_\e s_u \,, 
  &&
  \text{for all $i \leq k$, $u \leq r$, and $\e \geq \Tilde\Psi^{\alpha}(\mathbf{1})(t_i,s_u)$} \,,
  \label{eq:iteration}
\end{align}
then, by Lemma~\ref{lem:uniquefix} and (\Arch), we deduce $\vdash t_i \equiv_\e s_u$, for all $\e \geq \dist(t_i, s_u)$. Since $\vdash t \equiv_0 t_1$, $\vdash s \equiv_0 s_1$, by (\Triang), we deduce 
$\vdash t \equiv_\e s$, for all $\e \geq \dist(t,s)$, concluding the thesis.

\medskip
The reminder of the proof is devoted to prove \eqref{eq:iteration}. We do it 
by induction on $\alpha \in \naturals$.

(Base case: $\alpha = 0$) $\Tilde\Psi^0(\mathbf{1})(t_i,s_u) = \mathbf{1}(t_i,s_u)$. Since $\mathbf{1}(t_i,s_u) = 0$ whenever $t_i = s_u$ and $\mathbf{1}(t_i,s_u) = 1$ if  $t_i \neq s_u$, then $\eqref{eq:iteration}$ follows by the axioms (\Refl), (\Top) and (\Max).

(Inductive step: $\alpha > 0$). 
Recall that, by definition of $\Tilde{\Psi}$, we have the following:
\begin{equation*}
  \Tilde\Psi^{\alpha}(\mathbf{1})(t_i,s_u) = 
  \Tilde{\Psi}(\Tilde\Psi^{\alpha-1})(t_i,s_u) =
  \begin{cases}
    0 &\text{if $t_i \bisim_\UU s_u$,} \\
    \Psi(\Tilde\Psi^{\alpha-1})(t_i,s_u) &\text{otherwise} \,.
  \end{cases}
\end{equation*}
We consider the cases $t_i \bisim_\UU s_u$ and $t_i \not\bisim_\UU s_u$ separately.

Assume $t_i \bisim_\UU s_u$. Since our deduction system includes the one of Stark and Smolka, whenever $t_i \bisim_\UU s_u$, by completeness w.r.t.\ $\bisim_\UU$ (\cite[Theorem 3]{StarkS00}), we obtain $\vdash t_i \equiv_0 s_u$. By (\Max), $\vdash t_i \equiv_\e s_u$, for all $\e \geq \Tilde\Psi^{\alpha+1}(\mathbf{1})(t_i,s_u) = 0$.

Assume $t_i \not\bisim_\UU s_u$. 
Let $H,G$ be the formal sums on the right-hand side of \eqref{eq:tterms}, \eqref{eq:sterms}, 
respectively. Then, by definition of $\tauTT^*$, we have that, for $x,y \in (\Labels \times \TT{}) \uplus \X$,
\begin{align*}
  \tauTT^*(H)(x) = 
  \begin{cases}
  p_{ij}  & \text{if $\gamma(x) = t'_{ij}$, $j \leq h(i)$} \\
  0 &\text{otherwise} \,,
  \end{cases}
  &&
  \tauTT^*(G)(y) = 
  \begin{cases}
  q_{uv}  & \text{if $\gamma(y) = s'_{uv}$, $v \leq n(u)$} \\
  0 &\text{otherwise} \,,
  \end{cases}
\end{align*}
where $\gamma$ is the mapping such that
for all $t \in \TT{}$, $a \in \Labels$, and $X \in \X$,
\begin{align*}
\gamma((a, t)) = \pref{a}{t} \,, 
&& 
\gamma(X) = X \,,
&&
\gamma(\bot) = \rec{Z}{Z} \,.
\end{align*}
If we can prove that, for all $i \leq k$, $j \leq h(i)$, $u \leq r$, and $v \leq n(u)$,
\begin{align}
  \vdash t'_{ij} \equiv_\e s'_{uv} \,,
  &&
  \text{for all $\e \geq \Lambda(\Tilde\Psi^{\alpha-1}(\mathbf{1}))(\gamma(t'_{ij}),\gamma(s'_{uv}))$} \,,
  \label{eq:deduceLambda}
\end{align}
then, by Lemma~\ref{lem:deduceKantorovich}, we deduce 
\begin{align}
\vdash H \equiv_\e G \,,
&&
\text{for all $\e \geq \K[\Lambda(\Tilde\Psi^{\alpha-1}(\mathbf{1}))]{\tauTT^*(H),\tauTT^*(G)}$ } \,.
\label{eq:distHG}
\end{align}
Note that, by $t_i \not\bisim_\UU s_u$, \eqref{eq:tterms}, \eqref{eq:sterms}, and soundness 
of $\equiv_0$ w.r.t.\ $\bisim_\UU$ in \cite{StarkS00}, we have that 
$t_i \bisim_\UU H$, $s_u \bisim_\UU G$ and $H \not\bisim_\UU G$.
Therefore, by definition of $\Psi$ and triangular inequality
\begin{equation*}
\K[\Lambda(\Tilde\Psi^{\alpha-1}(\mathbf{1}))]{\tauTT^*(H),\tauTT^*(G)}
= \Tilde\Psi^{\alpha}(\mathbf{1})(H,G) 
= \Tilde\Psi^{\alpha}(\mathbf{1})(t_i,s_u) \,.
\end{equation*}
From \eqref{eq:distHG}, \eqref{eq:tterms}, \eqref{eq:sterms}, (\Triang), and the equality above we conclude \eqref{eq:iteration}.

\medskip
Next we prove \eqref{eq:deduceLambda}. 
The only interesting case is $t'_{i,j} = \pref{a}{t_{f(i,j)}}$ and $s'_{u,v} = \pref{a}{s_{z(u,v)}}$
---the others follow by using (\Refl), if $t'_{i,j} = s'_{u,v}$, (\Top) otherwise, and then (\Max).
By definition of $\gamma$ and $\Lambda$, we have 
$\Lambda(\Tilde\Psi^{\alpha-1}(\mathbf{1}))(\gamma(t'_{ij}),\gamma(s'_{uv})) = 
\Tilde\Psi^{\alpha-1}(\mathbf{1})(t_{f(i,j)}, s_{z(u,v)})$.
Now note that, by inductive hypothesis on $\alpha-1$, the following is deducible:
\begin{align*}
  \vdash t_{f(i,j)} \equiv_\e s_{z(u,v)} \,, 
  && 
  \text{for all $\e \geq \Tilde\Psi^\alpha(\mathbf{1})(t_{f(i,j)},s_{z(u,v)})$} \,,
\end{align*}
Therefore, \eqref{eq:deduceLambda} follows by the above and (\Pref).
\end{proof}

We conclude the section by showing a concrete example of deduction 
of the bisimilarity distance between two terms.
\begin{exa}
Consider the terms $t = \rec{X}{(\pref{a}{X} +_{\frac{1}{2}} Z)}$ and $s = \rec{Y}{(\pref{a}{Y} +_{\frac{1}{3}} Z)}$.
Similarly to Example~\ref{ex:semanticsTerms}, their pointed open Markov chain semantics are
\def\skiph{1}
\def\skipv{1.3}
\begin{align*}
\denot{t} \quad\bisim\quad
\tikz[labels, baseline={(current bounding box.center)}]{ 
  \draw (0,0) node[state, initial] (m) {$m$}; 
  \draw ($(m)+(down:\skipv)$) node[vars] (zvar) {$Z$};
  \path[-latex, font=\scriptsize]
    (m) edge node[right] {$\frac{1}{2}$} (zvar)
    (m) edge[loop right] node[right] {$a,\frac{1}{2}$} (m);
}
&&
\denot{s} \quad\bisim\quad
\tikz[labels, baseline={(current bounding box.center)}]{ 
  \draw (0,0) node[state, initial] (m) {$n$}; 
  \draw ($(m)+(down:\skipv)$) node[vars] (zvar) {$Z$};
  \path[-latex, font=\scriptsize]
    (m) edge node[right] {$\frac{2}{3}$} (zvar)
    (m) edge[loop right] node[right] {$a,\frac{1}{3}$} (m);
}
\end{align*} 
As shown in Example~\ref{ex:bisimidist} their bisimilarity distance is $\dist(t,s) = \frac{1}{4}$, hence
$\models t \equiv_\frac{1}{4} s$. 

Next we show how $\vdash t \equiv_\frac{1}{4} s$ can be deduced
by applying the axioms and rules of the quantitative deduction system proposed in Section~\ref{sec:quantAxioms}. For the sake of readability, the classical logical deduction rules (\Subst),
(\Cut), (\Assum) will be used implicitly, as well as (\Refl), (\Symm), and (\Triang).
By (\Unfold) and (\Fix) we can deduce 
\begin{align}
 \vdash t \equiv_0 \pref{a}{t} +_{\frac{1}{2}} Z
 &&\text{and}&&
 \vdash s \equiv_0 \pref{a}{s} +_{\frac{1}{3}} Z \,.
 \label{eq:EQNF}
\end{align}
Note that, $t$ and $s$ are now in the equational normal form of Theorem~\ref{th:eqcharacterization}.
The next step consists in applying (\Conv) to get information about the distance between the two 
probabilistic sums on the right-hand sides of the quantitative equations in \eqref{eq:EQNF}.
To do so, we first need to ``rearrange'' the sums in a way such that (\Conv) can actually be applied:
\begin{align}
\begin{aligned}
  \vdash \pref{a}{t} +_{\frac{1}{2}} Z  
  	&\equiv_0 (\pref{a}{t} +_{\frac{1}{3}} \pref{a}{t}) +_{\frac{1}{2}} Z  &&\text{(\Btwo)} \\
	&\equiv_0 \pref{a}{t} +_{\frac{1}{6}} (\pref{a}{t} +_{\frac{2}{5}} Z)  &&\text{(\SA)} \,,
\end{aligned}
&&
\begin{aligned}
  \vdash \pref{a}{s} +_{\frac{1}{3}} Z
  	&\equiv_0 Z +_{\frac{2}{3}} \pref{a}{s} &&\text{(\SC)} \\
  	&\equiv_0 (Z +_{\frac{1}{4}} Z) +_{\frac{2}{3}} \pref{a}{s}  &&\text{(\Btwo)} \\
	&\equiv_0 Z +_{\frac{1}{6}} (Z +_{\frac{3}{5}} \pref{a}{s})  &&\text{(\SA)}  \\
	&\equiv_0 Z +_{\frac{1}{6}} (\pref{a}{s} +_{\frac{2}{5}} Z)  &&\text{(\SC)} \,.  \\
\end{aligned}
\label{eq:KNF}
\end{align}
By (\Top), we deduce $\vdash \pref{a}{t} \equiv_1 Z$ and, by (\Refl), $\vdash Z \equiv_0 Z$. 
Hence, by (\Pref) and applying (\Conv) twice on $\pref{a}{t} +_{\frac{1}{6}} (\pref{a}{t} +_{\frac{2}{5}} Z)$
and $Z +_{\frac{1}{6}} (\pref{a}{s} +_{\frac{2}{5}} Z)$, we obtain the quantitative inference
\begin{equation}
  \{t \equiv_\e s \} \vdash 
  \pref{a}{t} +_{\frac{1}{6}} (\pref{a}{t} +_{\frac{2}{5}} Z)
  \equiv_{\frac{1}{3}\e + \frac{1}{6}}
  Z +_{\frac{1}{6}} (\pref{a}{s} +_{\frac{2}{5}} Z)
  \label{eq:KOper}
\end{equation}
Combining \eqref{eq:EQNF}, \eqref{eq:KNF}, \eqref{eq:KOper} we deduce the following
\begin{equation}
  \{t \equiv_\e s \} \vdash t \equiv_{\frac{1}{3}\e + \frac{1}{6}} s \,.
  \label{eq:psiOper}
\end{equation}
The above quantitative inference along with (\Top) $\vdash t \equiv_1 s$, can be thought of as 
a greatest fixed-point operator, taking an over approximation $\e \leq 1$ of the distance between 
$t$ and $s$ and refining it to $\frac{1}{3}\e + \frac{1}{6}$.
This interpretation is not bizarre, because it is exactly how the functional operator $\Tilde\Psi(d)$ operates 
on $t, s$ on a $1$-bounded pseudometric such that $d(t,s) = \e$.

The deduction of the distance $\dist(t,s) = \frac{1}{4}$ follows by proving that for all 
$1 \geq \delta > \frac{1}{4}$ we can deduce (in a finite number of steps!) $\vdash t \equiv_\delta s$ and then 
applying (\Arch).

The case $\delta = 1$ follows by (\Top). As for $1 > \delta > \frac{1}{4}$, notice that $T \colon [0,1] \to [0,1]$
defined as $T(\e) = \frac{1}{3}\e + \frac{1}{6}$ is a $\frac{1}{3}$-Lipschitz continuous map. Hence,
by Banach fixed-point theorem, $T$ has a unique fixed point, namely $\frac{1}{4}$, and the following inequality
holds for $q = \frac{1}{3}$
\begin{equation*}
  T^n(1) - \frac{1}{4} \leq \frac{q^n}{1- q} (T^0(1) - T^1(1)) \,.
\end{equation*}
The above reduces to $T^n(1) \leq \frac{1}{2} \cdot (\frac{1}{3})^{n-1}$. Hence by applying \eqref{eq:psiOper} 
$n$-times, starting from (\Top) $\vdash t \equiv_1 s$, for some integer $n \leq \log_3(\frac{1}{2 \delta}) + 1$, 
we deduce $\vdash t \equiv_{T^n(1)} s$ and we know that $T^n(1) \leq \delta$. Then, by (\Max), we 
deduce $\vdash t \equiv_{\delta} s$, \ie, the required deduction.
\qed
\end{exa}

\section{Axiomatization of the Discounted Bisimilarity Distance}
\label{sec:discountaxioms}

Next we describe how the deductive system in Section~\ref{sec:quantAxioms}
can be adapted to obtain soundness and completeness theorems w.r.t.\ the 
discounted bisimilarity distance of Deshainais et al.

\medskip
The distance $\dist$ that we considered so far is a special case (a.k.a. \emph{undiscounted}
bisimilarity distance) of the original definition by Deshainais et al.~\cite{DesharnaisGJP04}, 
which was parametric on a discount factor $\lambda \in (0,1]$. An equivalent definition of this distance
(due to van Breugel and Worrell~\cite{BreugelW:icalp01}) adapted to the case of open Markov chains 
is the following.
\begin{defi}[Discounted Bisimilarity Distance] \label{def:discbisimdist}
For $\lambda \in (0,1]$, the \emph{$\lambda$-discounted probabilistic bisimilarity pseudometric} 
$\dist[\M]^\lambda \colon M \times M \to [0,1]$ on $\M$ is the least fixed-point of the following functional 
operator on $1$-bounded pseudometrics (ordered point-wise),
\begin{equation*}
  \Psi_\M^\lambda(d)(m,m') =  \K[\Lambda^\lambda(d)]{\tau^*(m),\tau^*(m')}
  \tag{\sc $\lambda$-Kantorovich Operator}
\end{equation*}
where $\Lambda^\lambda(d)$ is the greatest 1-bounded pseudometric on $\big((\Labels \times M) \uplus \X \big)_\bot$ such that, for all $a \in \Labels$ and $t,s \in \TT{}$, $\Lambda(d)((a,t),(a,s)) = \lambda \cdot d(t,s)$.
\end{defi}

Clearly, for $\lambda = 1$ the above reduces to Definition~\ref{def:bisimdist}, hence $\dist = \dist^1$.
In~\cite[Theorem~6]{ChenBW12}, Chen et al.\ noticed that when $\lambda < 1$, $\Psi^\lambda$ is
a $\lambda$-Lipschitz continuous operator, \ie, for all $d,d' \colon M \times M \to [0,1]$, 
$\norm{\Psi^\lambda(d') - \Psi^\lambda(d)} \leq \lambda \norm{d' - d}$, where $\norm{f} = \sup_{x} |f(x)|$ is the supremum norm. So, by Banach fixed-point theorem $\dist^\lambda$ is the \emph{unique} fixed point
of $\Psi^\lambda$. 
For the same reason $\dist^\lambda = \bigsqcap_{\alpha \in \naturals} (\Psi^\lambda)^\alpha(\mathbf{1})$.
Moreover, the following also holds.
\begin{lem} \label{lem:discbisimkernel}
For any $\lambda < 1$, $\dist^\lambda(m,m') = 0$ iff $m \bisim m'$.
\end{lem}

\subsection{A Quantitative Deduction System for the Discounted Case}
\label{sec:discQDS}

In this section we provide a quantitative deduction system that is proved to be sound and complete
w.r.t.\ the $\lambda$-discounted probabilistic bisimilarity distance. 

The quantitative deduction system ${\vdash_\lambda} \subseteq 2^{\E[\Sigma]} \times \E[\Sigma]$
that we propose contains the one presented in Section~\ref{sec:quantAxioms},
where we add the following axiom
\begin{align*}
(\dPref)\qquad & \{ t \equiv_\e s \} \vdash_\lambda \pref{a}{t} \equiv_{\e'} \pref{a}{s} \,, \; 
\text{ for $\e' \geq \lambda \e$} \,.
\end{align*}
Notice that, when $\lambda < 1$, (\dPref) and (\Max) imply (\Pref) ---hence one may 
remove (\Pref) from the definition, since is redundant.

\medskip
The proof of soundness and completeness follow essentially in the same way of
Theorems~\ref{th:soundness} and \ref{th:completeness}. 
In the reminder of the section we only highlight the parts where some adjustments 
are needed.

\medskip
Notice that due to Lemma~\ref{lem:discbisimkernel} a similar result to 
Theorem~\ref{th:universaldist} holds also in the discounted case, \ie,
for all $t,s \in \TT{}$, $\dist[\OMC]^\lambda(\denot{t},\denot{s}) = \dist[\UU]^\lambda(t,s)$. 
So that, as done previously, we will use $\dist[\OMC]^\lambda$ and 
$\dist[\UU]^\lambda$, interchangebly, often simply denoted as $\dist^\lambda$.
Similarly, $\models_\lambda t \equiv_\e s$ will stand for $\dist^\lambda \leq \e$.

\begin{thm}[$\lambda$-Soundness] \label{th:discSoundness}
For arbitrary $t, s \in \TT{}$, if $\vdash_\lambda t \equiv_\e s$ then $\models_\lambda t \equiv_\e s$.
\end{thm}
\begin{proof}
The proof follows as Theorem~\ref{th:soundness}. We only need to check the soundness of (\dPref).
To do so it suffices to show that $\dist^\lambda(t,s) \geq \lambda \cdot \dist^\lambda(\pref{a}{t}, \pref{a}{s})$.
\begin{align*}
  \dist^\lambda(\pref{a}{t}, \pref{a}{s})
  &= \K[\Lambda^\lambda(\dist^\lambda)]{\tauTT^*(\pref{a}{t}), \tauTT^*(\pref{a}{s})} 
  \tag{$\dist^\lambda$ fixed-point \& def.\ $\Psi^\lambda$} \\
  &= \K[\Lambda^\lambda(\dist^\lambda)]{\ind{\{(a,t)\}}, \ind{\{(a,s)\}}} \tag{def.\ $\tauTT$ \& $\PP_\UU$} \\
  &= \Lambda^\lambda(\dist^\lambda)((a,t), (a,s)) \tag{def.\ $\K[]{}$} \\
  &= \lambda \cdot \dist^\lambda(t,s) \,. \tag{def.\ $\Lambda^\lambda$}
\end{align*}
\end{proof}

For the proof of completeness note that, since the quantitative deduction system $\vdash_\lambda$ is a conservative extension of $\vdash$, we have that Theorems~\ref{th:uniquesolution}, 
\ref{th:eqcharacterization}, and Lemma~\ref{lem:deduceKantorovich} still hold.
Moreover, for $\lambda < 1$, the proof of completeness is somehow simplified, since we do not need 
to introduce a new operator $\Tilde{\Psi}^\lambda$ to obtain the convergence to the $\lambda$-discounted
bisimilarity distance from above. Indeed, as noticed before $\dist^\lambda = \bigsqcap_{\alpha \in \naturals} (\Psi^\lambda)^\alpha(\mathbf{1})$.
\begin{thm}[$\lambda$-Completeness]\label{th:discCompleteness}
For arbitrary $t,s \in \TT{}$, if $\models_\lambda t \equiv_\e s$, then $\vdash_\lambda t \equiv_\e s$.
\end{thm}
\begin{proof}
The proof follows as in Theorem~\ref{th:completeness}. The only edits needed in are (i) syntactically replace $\vdash$ with $\vdash_\lambda$; (ii) replacing $\Tilde{\Psi}$ with $\Psi^\lambda$ in the proof of \eqref{eq:iteration}; and (iii) applying (\dPref) in place of (\Pref) for proving \eqref{eq:deduceLambda}.
\end{proof}


\section{A Quantitative Kleene's Theorem for Open Markov Chains}
\label{sec:expressibility}

In this last section we give a ``quantitative Kleene's theorem'' for pointed open Markov chains.
Specifically, we show that any (finite) pointed open Markov chains $(\M,m)$ can be represented 
up to bisimilarity as a $\O$-term $t_{(\M,n)}$ and, 
\textit{vice versa}, for any $\O$-term $t$, there exist a (finite) pointed 
open Markov chain bisimilar to $\denot{t}$. This establishes a representability theorem for finite open
Markov chains similar to the celebrated Kleene's theorem~\cite{Kleene56}, stating the correspondence
between \emph{regular expressions} and \emph{deterministic finite automata} (DFAs) up to language equivalence.

Even more interestingly, we show that by endowing the set of $\O$-terms with the pseudometric freely-generated  
by the quantitative deduction system presented in Section~\ref{sec:quantAxioms} 
(in a way which will be made precise
later) we get that the correspondence stated above is \emph{metric invariant}. 
We think of this result as `quantitative extension' of a Kleene's representation theorem for finite open 
Markov chains.

\subsection{Representability} \label{sec:represent}
We show that the class of expressible open Markov chains corresponds up to bisimilarity to the class of 
finite open Markov chains. Note that the results in this section can be alternatively obtained as in~\cite{SilvaBBR11} by observing that open Markov chains are coalgebras of a quantitative functor.

\medskip
A pointed Markov chain $(\M,m)$ is said \emph{expressible} if there exists a term $t \in \TT{}$ such that
$\denot{t} \bisim (\M,m)$.  
The next result is a corollary of Theorems~\ref{th:universal}, \ref{th:uniquesolution}, and \ref{th:soundness}.
\begin{cor} \label{cor:expressibility}
If $(\M, m)$ is finite then it is expressible. 
\end{cor}
\begin{proof}
We have to show that there exists $t \in \TT{}$ such that $\denot{t} \bisim (\M, m)$. 
Since the set of states $M = \{m_1,\dots, m_k\}$ is finite and, for each $m_i \in M$, $\tau(m_i)$ is finitely supported, then the sets of unguarded names $\{Y^i_{1},\dots, Y^i_{h(i)} \} = supp(\tau(m_i)) \cap \X$ and labelled transitions 
$\{\alpha^i_1, \dots, \alpha^i_{l(i)} \} = supp(m_i) \cap (\Labels \times M)$ of $m_i$ are finite. Let us associate with each $\alpha^i_j$ a name $X^i_j$, for all $i \leq k$ and $j \leq l(i)$. For each $i \leq k$, we define the terms
\begin{equation*}
  \textstyle
  t_i = \sum_{j = 1}^{l(i)} \tau(m_i)(\alpha^i_j) \cdot \pref{a^i_j}{X^i_j} + \sum_{j = 1}^{h(i)} \tau(m_i)(Y^i_j) \cdot Y^i_j \,,
\end{equation*}
where $\alpha^i_j = (a^i_j, m^i_j)$, for all $i \leq k$ and $i \leq l(i)$.
By Theorem~\ref{th:uniquesolution}, for $i \leq k$, there exists terms $\ol{s^i} = (s^i_1,\dots, s^i_{l(i)})$ such that $\vdash s_i \equiv_0 t_i[\ol{s^i}/\ol{X^i}]$, so that by soundness (Theorem~\ref{th:soundness}), $\denot{s_i} \bisim \denot{t_i[\ol{s^i}/\ol{X^i}]}$. Hence, by Theorem~\ref{th:universal}, we have 
$(\UU, s_i) \bisim (\UU, t_i[\ol{s^i}/\ol{X^i}])$.

Let $\ol{m^i} = (m^i_1, \dots, m^i_{l(i)})$ and $\ol{X^i} = (X^i_1, \dots, X^i_{l(i)})$, for $i \leq k$. It is a routine check to prove that the smallest equivalence relation $R_i$ containing $\{ (m_i, t_i[\ol{m^i}/\ol{X^i}]) \mid i \leq k \}$ is a bisimulation for $(\M, m_i)$ and $(\UU(\M), t_i[\ol{m^i}/\ol{X^i}])$, hence $(\M, m_i) \bisim (\UU(\M), t_i[\ol{m^i}/\ol{X^i}])$.
Similarly, one can prove $(\UU(\M), t_i[\ol{m^i}/\ol{X^i}]) \bisim (\UU, t_i[\ol{s^i}/\ol{X^i}])$ by taking the smallest equivalence relation containing $\{ (t_i[\ol{m^i}/\ol{X^i}], t_i[\ol{s^i}/\ol{X^i}]) \mid i \leq k \}$ and 
$\{ (m^i_j, s^i_j) \mid i \leq k, j \leq l(i) \}$. By transitivity of $\bisim$, $(\M, m_i) \bisim \denot{s_i}$, for all 
$i \leq k$, hence $(\M, m)$ is expressible.
\end{proof}

The converse (up to bisimilarity) of the above result can also be proved, and it follows as a corollary of 
Theorems~\ref{th:universal}, \ref{th:soundness}, and \ref{th:eqcharacterization}.
\begin{cor} \label{cor:finiteexpressibility}
If $(\M,n)$ is expressible then it is finite up-to-bisimilarity.
\end{cor}
\begin{proof}
Let $t \in \TT{}$. We have to show that there exists $(\M,m) \in \OMC$ with a finite set of states such that $\denot{t} \bisim (\M, m)$. From Theorem~\ref{th:eqcharacterization}, there exist $t_1, \ldots, t_k$ with free names in $\ol{Y}$, such that
$\vdash t \equiv_0 t_1$ and 
\begin{align*}
  \textstyle
  \vdash t_i \equiv_0 \sum_{j=1}^{h(i)} p_{ij} \cdot s_{ij} + \sum_{j=1}^{l(i)} q_{ij} \cdot Y_{g(i,j)} \,,
  && \text{for all $i \leq k$,}
\end{align*}
where the terms $s_{ij}$ and names $Y_{g(i,j)}$ are enumerated without repetitions, and $s_{ij}$ is either $\rec{X}{X}$ or has the form $\pref{a_{ij}}{t_{f(i,j)}}$. Let $Z_1, \dots, Z_k$ be fresh names distinct from $\ol{Y}$, and define $t'_i$ as the term obtained by replacing in the right end side of the equation above each occurrence of $t_i$ with $Z_i$.
Then, clearly $\vdash t_i \equiv_0 t'_i[\ol{t}/\ol{Z}]$. By soundness (Theorem~\ref{th:soundness}), we have that
$\denot{t_i} \bisim \denot{t'_i[\ol{t}/\ol{Z}]}$, so that, by Theorem~\ref{th:universal}, 
$(\UU, t_i) \bisim (\UU, t'_i[\ol{t}/\ol{Z}])$.

Define $\M= (M,\tau)$ by setting $M = \{t_1, \dots, t_k\}$, $m = t_1$, and, for all $i \leq k$,
taking as $\tau(t_i)$ the smallest sub-probability distribution on $(\Labels \times M) \uplus \X$ such that 
$\tau(t_i)((a_{ij}, t_{f(i,j)})) = p_{ij}$ and $\tau(t_i)(Y_{g(i,e)}) = q_{ie}$, for all $i \leq k$, $j \leq h(i)$, and $e \leq l(i)$.
Notice that since the equation above is without repetitions, $\tau$ is well defined. Moreover, 
$1-\tau(m_i)((\Labels \times M) \uplus \X) = p_{iw}$ whenever there exists $w \leq h(i)$ such that $s_{iw} = \rec{X}{X}$.
It is not difficult to prove that $(\M, t_i) \bisim (\UU, t'_i[\ol{t}/\ol{Z}])$ (take the smallest equivalence relation containing the pairs $(t_i, t'_i[\ol{t}/\ol{Z}])$, for $i \leq k$), so that by transitivity of $\bisim$, $(\M, t_i) \bisim \denot{t_i}$, for all 
$i \leq k$. By $\vdash t \equiv_0 t_1$ and Theorem~\ref{th:soundness}, we also have $\denot{t} \bisim \denot{t_1}$, thus $\denot{t} \bisim (\M, m)$.
\end{proof}

\subsection{A Quantitative Kleene's Theorem}
\label{sec:qKleene}
We provide a metric analogue to Kleene's representation theorem for finite pointed open Markov chains.

\medskip
In~\cite{MardarePP:LICS16}, Mardare et al.\ gave a construction for the free model
of a quantitative theory of a generic quantitative deduction system. 
Here we present their definition only for the specific case of the quantitative deduction system 
presented in Section~\ref{sec:quantAxioms}. Note that since our quantitative deduction system 
does not require all algebraic operators to be non-expansive, by applying this construction 
we obtain a relaxed quantitative algebra, not a proper one in the of sense 
of~\cite{MardarePP:LICS16} . 

\begin{defi}[Initial $\vdash$-model]
The \emph{initial $\vdash$-model} is defined as the relaxed quantitative algebra $(\TT{},\O,d_{\TT{}})$,
where $(\TT{},\O)$ is the initial algebra of $\O$-terms and $d_{\TT{}}$
is the $1$-bounded pseudometric on $\TT{}$ defined, for arbitrary terms $t,s \in \TT{}$, as
$d_{\TT{}}(t,s) = \inf \set{\e}{ \; \vdash t \equiv_\e s}$.
\end{defi}

Note that by (\Refl), (\Symm), (\Triang), (\Top) it is easy to prove that $d_{\TT{}}$ is a well-defined 
$1$-bounded pseudometric. Moreover, $(\TT{},\O,d_{\TT{}})$ is clearly a sound model for $\vdash$.

The attractiveness of the above model, as opposed to an operational one looking directly at the 
operational semantics of terms, is that it is purely equational. Indeed, one can reason about its 
properties by just proving statements about the distance between terms using classical equational deduction in the system $\vdash$.

Next we show that there is a strong correspondence between the initial $\vdash$-model and
the quantitative algebra of finite pointed open Markov chains.
\begin{thm}[Quantitative Kleene's Theorem] \label{th:kleenetheorem} \
\begin{enumerate}[label={(\roman*)}, itemsep=0.7ex]

  \item \label{kleene1}
  For every pair $(\M,m), (\N,n)$ of finite pointed open Markov chains, there exist 
  terms $t, s \in \TT{}$ such that $\denot{t} \bisim (\M,m)$, $\denot{s} \bisim (\N,n)$, and 
  $\dist((\M,m), (\N,n)) = d_{\TT{}}(t, s)$;
  
  \item \label{kleene2}
  for every pair $t,s \in \TT{}$, there exist finite pointed open Markov chains
  $(\M,m), (\N,n)$, such that $\denot{t} \bisim (\M,m)$, $\denot{s} \bisim (\N,n)$, and 
  $\dist((\M,m), (\N,n)) = d_{\TT{}}(t, s)$.
\end{enumerate}
\end{thm}
\begin{proof}
Before proving \ref{kleene1} and \ref{kleene2}, we show that, for all $t,s \in \TT{}$ 
\begin{equation}
d_{\TT{}}(t,s) = \dist(\denot{t},\denot{s}) \,.
\label{eq:termdistIsBisim}
\end{equation}
The equality follows trivially by definition of $d_{\TT{}}$ and the soundness and completeness theorems
(Theorems~\ref{th:soundness} and \ref{th:soundness}). Indeed,
\begin{equation*}
d_{\TT{}}(t,s) 
= \inf \set{\e}{ \; \vdash t \equiv_\e s}
= \inf \set{\e}{ \; \models t \equiv_\e s}
= \inf \set{\e}{ \dist(\denot{t},\denot{s}) \leq \e }
= \dist(\denot{t},\denot{s}) \,.
\end{equation*}

\ref{kleene1} Given $(\M,m), (\N,n)$ finite pointed open Markov chains, we construct the $\O$-terms $t, s$ 
as in Corollary~\ref{cor:expressibility}, obtaining that $\denot{t} \bisim (\M,m)$, $\denot{s} \bisim (\N,n)$.
Now the results follows by \eqref{eq:termdistIsBisim} and Theorem~\ref{th:universaldist}.

\ref{kleene2} Given $t,s \in \TT{}$, we construct the finite pointed Markov chains $(\M,m), (\N,n)$ as in
Corollary~\ref{cor:finiteexpressibility}, obtaining that $\denot{t} \bisim (\M,m)$, $\denot{s} \bisim (\N,n)$.
Again the results follows by \eqref{eq:termdistIsBisim} and Theorem~\ref{th:universaldist}.
\end{proof}

\begin{rem}[The discounted case]
Note that, by using the quantitative deduction system $\vdash_\lambda$ presented in Section~\ref{sec:discountaxioms}, 
adjusting the proof of Theorem~\ref{th:kleenetheorem} in the obvious way, a quantitative Kleene's 
theorem can be obtained also for the $\lambda$-discounted bisimilarity distance. 
\end{rem}

\section{Conclusions and Future Work}
\label{sec:conclusion}

In this paper we proposed a sound and complete axiomatization for the bisimilarity distance of Desharnais et al., later extended to its discounted variant. 
The axiomatic system proposed comes as a natural generalization of Stark and Smolka's one~\cite{StarkS00} for probabilistic bisimilarity, to which we added only three extra axioms, namely (\Top), (\Pref), (\Conv)
---along those required by the quantitative equational framework.

Although the use of the recursion operator does not fit the general framework of Mardare et al.~\cite{MardarePP:LICS16}, we were able to prove completeness in a way that we believe is general 
enough to accommodate the axiomatization of other behavioral distances for probabilistic systems.
A concrete example of this statement is provided in~\cite{BacciLMmfps17}, where we axiomatized the total variation distance for Markov chains.
In the light of our result it would be nice to see how much of the work in~\cite{MardarePP:LICS16}
truly bases on the non-expansivity assumption for the algebraic operators. As as possible future work,
it would be interesting to extend the general quantitative framework to algebraic operators that are
only required to be Lipschitz-continuous (indeed, our proof uses the fact that the functional fixed point 
operator defining the recursion is $q$-Lipschitz continuous for some $q < 1$).

Another appealing direction of future work is to apply our results on quantitative systems
described as coalgebras in a way similar to one proposed in~\cite{SilvaBBR11,Bonsangue13}.
By pursuing this direction we would be able to obtain metric axiomatization for 
a wide variety of quantitative systems, including weighted transition systems, Segala's 
systems, stratified systems, Pnueli-Zuck systems, etc.

A very recent related work worth to be mentioned is~\cite{BacciMPP:LICS18}, where
Markov processes have been axiomatised via the standard framework of quantitative equational
theories via disjoint union of theories. The signature used there is very similar to the one presented in the present work, but they managed to obtain completeness with no need of a recursion operator. Completeness was obtained by taking as the complete model the Cauchy completion of the standard quantitative initial algebra ---intuitively, Cauchy completion introduces recursive behaviours as the limit of infinitely many unfolding operations.

\section*{Acknowledgments}
We thank the anonymous reviewers of CONCUR 2016
for the useful comments and suggestions. The first author would like 
to thank Frank van Breugel for the suggestions on possible extensions of this work.
One of them, namely the axiomatization of the discounted probabilistic bisimilarity distance,  
has been developed in this extended version; the others have been resolved in~\cite{BacciLMmfps17}.\nocite{BacciM15,BacciM12}

\bibliographystyle{alpha}
\bibliography{biblio}


%
%

%

\end{document}